\documentclass[acmsmall, screen, nonacm]{acmart} 
\setcopyright{none}

\usepackage{amsthm}
\usepackage{fancyvrb}
\usepackage{mathpartir}
\usepackage{mathtools}
\usepackage{microtype}
\usepackage{xcolor}
\usepackage{twemojis}
\usepackage{tikz-cd}
\usepackage{aliascnt}
\usepackage{trimclip}
\usepackage{wrapfig}
\usepackage{wasysym}
\usepackage{subdepth}

\definecolor{weakening-strong}{rgb}{0.984,0.777,0.414}
\definecolor{contraction-strong}{rgb}{0.316,0.621,0.688}
\definecolor{exchange-strong}{rgb}{0.539,0.371,0.617}
\definecolor{weakening}{rgb}{0.992,0.922,0.801}
\definecolor{contraction}{rgb}{0.859,0.922,0.934}
\definecolor{exchange}{rgb}{0.906,0.871,0.922}
\definecolor{error}{rgb}{0.89,0.00,0.00}
\definecolor{hblue}{HTML}{205EFF}
\definecolor{hred}{HTML}{B41D67}

\definecolor{Gray}{HTML}{D5D6D8}
\definecolor{Green}{HTML}{61BB46}
\definecolor{Yellow}{HTML}{FFB000}
\definecolor{Orange}{HTML}{FE6100}
\definecolor{Red}{HTML}{DC267F}
\definecolor{Purple}{HTML}{785EF0}
\definecolor{Blue}{HTML}{648FFF}
\definecolor{Cyan}{HTML}{009DDC}
\definecolor{Black}{HTML}{3D3D3D}

\makeatletter
\DeclareFontFamily{OMX}{MnSymbolE}{}
\DeclareSymbolFont{MnLargeSymbols}{OMX}{MnSymbolE}{m}{n}
\SetSymbolFont{MnLargeSymbols}{bold}{OMX}{MnSymbolE}{b}{n}
\DeclareFontShape{OMX}{MnSymbolE}{m}{n}{
    <-6>  MnSymbolE5
   <6-7>  MnSymbolE6
   <7-8>  MnSymbolE7
   <8-9>  MnSymbolE8
   <9-10> MnSymbolE9
  <10-12> MnSymbolE10
  <12->   MnSymbolE12
}{}
\DeclareFontShape{OMX}{MnSymbolE}{b}{n}{
    <-6>  MnSymbolE-Bold5
   <6-7>  MnSymbolE-Bold6
   <7-8>  MnSymbolE-Bold7
   <8-9>  MnSymbolE-Bold8
   <9-10> MnSymbolE-Bold9
  <10-12> MnSymbolE-Bold10
  <12->   MnSymbolE-Bold12
}{}

\let\llangle\@undefined
\let\rrangle\@undefined
\DeclareMathDelimiter{\llangle}{\mathopen}
                     {MnLargeSymbols}{'164}{MnLargeSymbols}{'164}
\DeclareMathDelimiter{\rrangle}{\mathclose}
                     {MnLargeSymbols}{'171}{MnLargeSymbols}{'171}
\makeatother

\newcommand{\error}[1]{\textcolor{error}{#1}}
\newcommand{\danger}{\error{\raisebox{0.1em}{\fontencoding{U}\fontfamily{futs}\selectfont\char 49\relax}}}

\newcommand{\hole}{\fbox{\small \ensuremath{\mathsf{?}}}}

\newcommand{\valpha}{\ensuremath{\alpha}}
\newcommand{\vbeta}{\ensuremath{\beta}}
\newcommand{\vdelta}{\ensuremath{\delta}}

\input{pygments_style.tex}

\theoremstyle{} 
\newtheorem{theorem}{Theorem}[section]
\newcommand{\addtheorem}[2]{
  \newaliascnt{#1}{theorem}
  \newtheorem{#1}[#1]{#2}
  \aliascntresetthe{#1}
  \ExpandArgs{c}\newcommand{#1autorefname}{#2}
}
\addtheorem{lemma}{Lemma}
\addtheorem{lemmav}{[V] Lemma}
\addtheorem{lemmac}{[C] Lemma}
\addtheorem{corollary}{Corollary}
\addtheorem{property}{Property}

\newcommand{\Sub}[3]{\ensuremath{[#1/#2]#3}}

\newcommand{\bnfsep}{\mathrel{|}}

\newcommand{\evals}{\longmapsto}

\newcommand{\exact}[5]{#1 \sim #2 \in #3 \mathrel{|} #5; #4}
\newcommand{\exactOpen}[8]{#1; #2; #3 \gg #4 \sim #5 \in #6 \mathrel{|} #8; #7}

\newcommand{\have}{\mathrel{.}}

\newcommand{\typing}[5]{\mbox{\ensuremath{#1; #2; #3 \vdash #4 : #5}}}

\newcommand{\defmathsymb}[1]{
  \mathchoice
  {\mbox{\normalshape #1}}
  {\mbox{\normalshape #1}}
  {\scalebox{.7}{\normalshape #1}}
  {\scalebox{.7}{\normalshape #1}}
}

\newcommand{\ul}[1]{\underline{#1}}
\newcommand{\open}{{\defmathsymb{\fullmoon}}}
\newcommand{\opTy}[2]{\ensuremath{\defmathsymb{\fullmoon}_{#1}(#2)}}
\newcommand{\closed}{{\defmathsymb{\newmoon}}}
\newcommand{\clTy}[2]{\ensuremath{\defmathsymb{\newmoon}_{#1}(#2)}}
\newcommand{\unitTy}{\texttt{unit}}
\newcommand{\arrTy}[2]{\ensuremath{#1 \rightarrow #2}}

\newcommand{\forallTy}[2]{\ensuremath{\forall(#1 . #2)}}
\newcommand{\sumTy}[2]{\ensuremath{#1 + #2}}
\newcommand{\existsTy}[2]{\ensuremath{\exists(#1 . #2)}}
\newcommand{\uTy}[1]{\ensuremath{\mathsf{U}(#1)}}
\newcommand{\fTy}[1]{\ensuremath{\mathsf{F}(#1)}}
\newcommand{\pairTy}[2]{\ensuremath{#1 \otimes #2}}

\newcommand{\unitEx}{\ensuremath{\langle\rangle}}
\newcommand{\sealEx}[2]{\ensuremath{\texttt{seal}[#1](#2)}}
\newcommand{\unsealEx}[3]{\ensuremath{\texttt{unseal}\; #1\; \texttt{in}\; #2 . #3}}
\newcommand{\suspEx}[1]{\ensuremath{\texttt{susp}(#1)}}
\newcommand{\forceEx}[1]{\ensuremath{\texttt{force}(#1)}}
\newcommand{\packEx}[2]{\texttt{pack}[#1](#2)}
\newcommand{\openEx}[4]{\texttt{open}\; #1\; \texttt{in}\; #2 , #3 . #4}
\newcommand{\lamEx}[2]{\ensuremath{\lambda(#1 . #2)}}
\newcommand{\apEx}[2]{\texttt{ap}(#1 ; #2)}
\newcommand{\allEx}[2]{\ensuremath{\Lambda(#1 . #2)}}
\newcommand{\instEx}[2]{\ensuremath{#1 [#2]}}
\newcommand{\injlEx}[1]{\ensuremath{\texttt{l} \cdot #1}}
\newcommand{\injrEx}[1]{\ensuremath{\texttt{r} \cdot #1}}
\newcommand{\caseEx}[5]{\ensuremath{\texttt{case}\; #1\; \texttt{of}\; \{\; \injlEx{#2} \hookrightarrow #3 \mathrel{|} \injrEx{#4} \hookrightarrow #5\; \}}}
\newcommand{\pairEx}[2]{\ensuremath{\langle #1 , #2 \rangle}}
\newcommand{\splitEx}[4]{\texttt{split}\; #1\; \texttt{into}\; \pairEx{#2}{#3}\; \texttt{in}\; #4}
\newcommand{\retEx}[1]{\ensuremath{\texttt{ret}(#1)}}
\newcommand{\bindEx}[3]{\ensuremath{\texttt{bind}\; #1\; \texttt{in}\; #2 . #3}}
\newcommand{\consumeEx}[1]{\ensuremath{\texttt{consume}(#1)}}
\newcommand{\produceEx}[1]{\ensuremath{\texttt{produce}(#1)}}

\begin{document}

\title{The Duality of Information Flow} 
\subtitle{Reconciling Robust Downgrading with Non-Interference}

\author{Hemant Gouni}
\orcid{0009-0009-3888-8440}
\affiliation{
  \institution{Carnegie Mellon University}
  \city{Pittsburgh}
  \country{USA}
}
\email{hsgouni@cs.cmu.edu}

\author{Frank Pfenning}
\orcid{0000-0002-8279-5817}
\affiliation{
  \institution{Carnegie Mellon University}
  \city{Pittsburgh}
  \country{USA}
}
\email{fp@cs.cmu.edu}

\author{Jonathan Aldrich}
\orcid{0000-0003-0631-5591}
\affiliation{
  \institution{Carnegie Mellon University}
  \city{Pittsburgh}
  \country{USA}
}
\email{jonathan.aldrich@cs.cmu.edu}

\begin{abstract}
\textit{Non-interference} properties have long enjoyed a position as the
high water mark of program security guarantees. For instance,
\textit{confidentiality} states that a secret does not interfere with a
given computation, precluding the former from escaping through the latter.
\textit{Integrity} analogously states that untrusted data does not
interfere with some computation, preventing the former from manipulating
the latter. \textit{Information flow type systems} comprise the primary
means for obtaining non-interference properties of programs, but their
potential as a holy grail for secure programming has remained latent. Prior
work bifurcates the type system along confidentiality and integrity,
resulting in duplicate reasoning machinery and complex specifications.
Furthermore, long-held wisdom dictates that non-interference must be
weakened with \textit{downgrading} mechanisms to accommodate the needs
of practical programs, nearly all of which violate confidentiality and
integrity in the course of fulfilling their purpose. This often pierces
abstraction barriers and compromises modular reasoning.

We introduce \textit{parametric information flow}, which uses recent
insights from modal type theory to shed light on these issues. In
particular, we draw inspiration from work on the \textit{open} and
\textit{closed} modalities, highlighting their rich interplay. Though each
individual modality finds uses throughout the literature, our key insight
is that \textit{their joint interaction} suffices to reconstruct
full-spectrum information flow reasoning, producing a single framework
accounting for both confidentiality and integrity. Downgrading and
analogues of advanced reasoning tools in the lineage of \textit{robust
declassification} are recovered without extensions to our theory,
strengthening prior results. We show non-interference via a binary logical
relations argument, realizing robustness as an ordinary 2-hyperproperty
mediated by our modalities. Our work reveals state-of-the-art
downgrading mechanisms to be wholly compatible with those for
abstraction and modularity, arising precisely from the semantics of the
latter under full-strength non-interference.
\end{abstract}

\maketitle

\section{Introduction}
\label{sec:intro}

Shall secure information flow focus on tracking \textit{information} or
\textit{flows}? That is, in considering the dependency of a computation on some
data, should primacy be given to what can be \textit{gleaned about} that data
by the computation, or to the manner of its \textit{movement through} the
computation? We say both. Much prior work in information flow chooses one or
the other, a seemingly immaterial choice between options so subtly distinct as
to be without a difference. However, we will show that it makes all the
difference in the world to consider them as equal parts of a whole.

\subsection{Information Flow by Parametric Polymorphism}

In embarking on illuminating the preceding remark---a process which will span
the length of this work---we begin with a quick tour of our approach to
information flow. We start at the familiar machinery of parametric polymorphism
\cite{reynolds_1984}. Consider the typing for the polymorphic identity
\mbox{\texttt{id} \texttt{\PY{o}{:}} \PY{n+nl}{\valpha} \texttt{\PY{o}{->}}
\PY{n+nl}{\valpha}}, which constrains its return value to be exactly its
argument. Likewise, the typing \mbox{\texttt{fst} \texttt{\PY{o}{:}}
\PY{n+nl}{\valpha} \texttt{\PY{o}{->}} \PY{n+nl}{\vbeta} \texttt{\PY{o}{->}}
\PY{n+nl}{\valpha}} specifies a function which returns its first argument and
ignores its second. Finally, the type \mbox{\PY{k+kt}{\texttt{list}}
\PY{n+nl}{\valpha} \PY{o}{\texttt{->}} \PY{o}{\texttt{(}}\PY{n+nl}{\valpha}
\PY{o}{\texttt{->}} \PY{n+nl}{\vbeta}\PY{o}{\texttt{)}} \PY{o}{\texttt{->}}
\PY{k+kt}{\texttt{list}} \PY{n+nl}{\vbeta}} of the \texttt{map} function on
lists states that the elements of the argument list are passed into the higher
order function, whose return values in turn populate the elements of the
returned list. In each case, \textit{parametricity} allows us to reason
simultaneously about properties of both the \textit{information} referenced by
each signature and the \textit{flows} described by it.

Specifically, we reason about flows by interpreting each type variable as
naming some data, with its positioning inside the signature explicating the
movement of that data through the computation. In the case of
\mbox{\PY{n+nl}{\valpha} \texttt{\PY{o}{->}} \PY{n+nl}{\valpha}} the appearance
of \PY{n+nl}{\valpha} in the return type tells us that the computation
\textit{depends on} the argument at \PY{n+nl}{\valpha}. On the other hand, we
reason about information through data abstraction, or the hiding of
information. This is what lets us know that a computation at type
\mbox{\PY{n+nl}{\valpha} \texttt{\PY{o}{->}} \PY{n+nl}{\valpha}} \textit{cannot
depend on} the details of the particular \PY{n+nl}{\valpha} passed to it. Taken
together, these two properties of data dependency---or information
flow---narrow the space of possible behaviors of a computation at type
\mbox{\PY{n+nl}{\valpha} \texttt{\PY{o}{->}} \PY{n+nl}{\valpha}} to contain
only the identity function, a classic result of parametricity; similar results
can be obtained for \texttt{fst} and \texttt{map}. Usual appeals to this
principle merge both threads of reasoning into one.

The unification by parametricity of concerns about both flows and information
is often quite useful, but its sheer strength significantly constrains the
range of expressible programs. For instance, what if we wish to speak of the
information flow properties of a function \mbox{\texttt{increment}
\PY{o}{\texttt{:}} \PY{k+kt}{\texttt{int}} \PY{o}{\texttt{->}}
\PY{k+kt}{\texttt{int}}} which computes with---so cannot be polymorphic
in---its input? Remedying this necessitates distinguished notions of flows and
information able to be deployed separately. We proceed to isolate each in turn,
making it available in a more general setting. We initially retrace the
development by \citet{gouni2025structural} of certain forms of information flow
via parametric polymorphism.

\subsubsection{Isolating Flow Reasoning}
\label{sec:isolating-flow-intro}

The preceding approach to reasoning about the movement of data within some
signature was to trace its associated type variables \PY{n+nl}{\valpha}, but
this will not work if we also wish to perform arbitrary computations with that
data. Let us instead try \textit{tagging} types with \textit{dependency
variables} like \PY{n+nl}{\valpha}, rather than having \PY{n+nl}{\valpha}
\textit{be} the type. Under this proposal, we might type the increment function
on integers as \mbox{\PY{n+nl}{\valpha} \PY{k+kt}{\texttt{int}}
\PY{o}{\texttt{->}} \PY{n+nl}{\valpha} \PY{k+kt}{\texttt{int}}}. The
\PY{n+nl}{\valpha} names the attached data, allowing us to track its movement
from the argument at \PY{n+nl}{\valpha} to the return value dependent on
\PY{n+nl}{\valpha}. Importantly, unlike for type variables, \PY{n+nl}{\valpha}
no longer plays any role in reasoning about information hiding.

This seems to leave us better off than before, in that we can simultaneously
track the movement of data and compute with it, but it is still impractical. To
see why, imagine how we might type addition on integers. We would like to
express that it takes two arguments and returns a result dependent on both. We
start by writing \mbox{\PY{n+nl}{\valpha} \PY{k+kt}{\texttt{int}}
\PY{o}{\texttt{->}} \PY{n+nl}{\vbeta} \PY{k+kt}{\texttt{int}}
\PY{o}{\texttt{->}} \hole{} \PY{k+kt}{\texttt{int}}}. How should the \hole{} be
filled? Our current syntax limits us to mentioning a single dependency variable
inside it, but we clearly need to mention both \PY{n+nl}{\valpha} and
\PY{n+nl}{\vbeta} because both arguments flow to the return value. So we
generalize our syntax to handle \textit{sets} of dependency variables within
types, rather than just \textit{singular} dependency variables. Setting \hole{}
= \PY{o}{\texttt{[}}\PY{n+nl}{\valpha} \PY{n+nl}{\vbeta}\PY{o}{\texttt{]}}, we
successfully type addition as
\mbox{\PY{o}{\texttt{[}}\PY{n+nl}{\valpha}\PY{o}{\texttt{]}}
\PY{k+kt}{\texttt{int}} \PY{o}{\texttt{->}}
\PY{o}{\texttt{[}}\PY{n+nl}{\vbeta}\PY{o}{\texttt{]}} \PY{k+kt}{\texttt{int}}
\PY{o}{\texttt{->}} \PY{o}{\texttt{[}}\PY{n+nl}{\valpha}
\PY{n+nl}{\vbeta}\PY{o}{\texttt{]}} \PY{k+kt}{\texttt{int}}}. For consistency,
the dependency variables on the arguments have also been generalized to
dependency sets.

We refer to \PY{n+nl}{\valpha} as a dependency \textit{variable} because our
signatures are still polymorphic, but over dependencies rather than types. In
particular, each \PY{n+nl}{\valpha} ranges over another dependency set
\mbox{\PY{o}{\texttt{[}}\PY{n+nl}{$\alpha_1$} \ldots{}
\PY{n+nl}{$\alpha_n$}\PY{o}{\texttt{]}}}, meaning the former may be
instantiated to the latter. Just as the \PY{n+nl}{\valpha} in
\PY{n+nl}{\valpha} \PY{o}{\texttt{->}} \PY{n+nl}{\valpha} can be instantiated
to \PY{k+kt}{\texttt{int}} to obtain \PY{k+kt}{\texttt{int}}
\PY{o}{\texttt{->}} \PY{k+kt}{\texttt{int}}, we can instantiate the
\PY{n+nl}{\valpha} in our type for addition to
\mbox{\PY{o}{\texttt{[}}\PY{n+nl}{\texttt{secret1}}
\PY{n+nl}{\texttt{secret2}}\PY{o}{\texttt{]}}} to obtain
\mbox{\PY{o}{\texttt{[}}\PY{n+nl}{\texttt{secret1}}
\PY{n+nl}{\texttt{secret2}}\PY{o}{\texttt{]}} \PY{k+kt}{\texttt{int}}
\PY{o}{\texttt{->}} \PY{o}{\texttt{[}}\PY{n+nl}{\vbeta}\PY{o}{\texttt{]}}
\PY{k+kt}{\texttt{int}} \PY{o}{\texttt{->}}
\PY{o}{\texttt{[}}\PY{n+nl}{\texttt{secret1}} \PY{n+nl}{\texttt{secret2}}
\PY{n+nl}{\vbeta}\PY{o}{\texttt{]}} \PY{k+kt}{\texttt{int}}}. A function at
this type takes an integer at the dependencies
\PY{o}{\texttt{[}}\PY{n+nl}{\texttt{secret1}}
\PY{n+nl}{\texttt{secret2}}\PY{o}{\texttt{]}} and eventually returns an integer
at those dependencies plus those of its second argument. As a notational
convention, we will use \mbox{\PY{n+nl}{\valpha}, \PY{n+nl}{\vbeta},
\PY{n+nl}{\vdelta}, \ldots{}} to indicate polymorphic dependency variables; all
others are fixed.

For a more grounded example of tracking flows consider the case of a password
checker in \autoref{fig:password-check}. We first declare a
\PY{k+kt}{\texttt{string}} \texttt{pass} tagged with \PY{n+nl}{\texttt{secret}} to
indicate that it represents password data. We assume \PY{n+nl}{\texttt{secret}} to
be in scope here, but defer detailing the mechanisms used to introduce it.
\texttt{check} takes a password \texttt{attempt} \PY{k+kt}{\texttt{string}} at
any dependencies \PY{n+nl}{\valpha} and compares it to \texttt{pass}, returning
a \PY{k+kt}{\texttt{bool}} tagged with \PY{n+nl}{\valpha} and
\PY{n+nl}{\texttt{secret}} to indicate the flows induced from \texttt{attempt} and
\texttt{pass}.

\begin{figure}
    \begin{BVerbatim}[commandchars=\\\{\}]
\textcolor{lightgray}{1} \PY{k}{let} password \PY{o}{:} \PY{o}{[}\PY{n+nl}{secret}\PY{o}{]} \PY{k+kt}{string} \PY{o}{=} \PY{l+s}{\PYZdq{}takver\PYZdq{}}
\textcolor{lightgray}{2}
\textcolor{lightgray}{3} \PY{k}{let} check \PY{o}{:} \PY{o}{[}\PY{n+nl}{\valpha}\PY{o}{]} \PY{k+kt}{string} \PY{o}{->} \PY{o}{[}\PY{n+nl}{\valpha} \PY{n+nl}{secret}\PY{o}{]} \PY{k+kt}{bool} \PY{o}{=}
\textcolor{lightgray}{4}    \PY{k}{fun} attempt \PY{o}{->} attempt \PY{o}{==} password
    \end{BVerbatim}
    \caption{Simple Password Checker}
    \label{fig:password-check} 
\end{figure}

Note that the type of \texttt{check} is slightly redundant: the appearance of
\PY{n+nl}{\valpha} in both the argument and return type means any call to
\texttt{check} will \textit{always} depend on the argument passed to it. We can
use \textit{dependency elision} \cite{gouni2025structural} to simplify its
type, eliding \PY{n+nl}{\valpha} to get \mbox{\texttt{check} \PY{o}{\texttt{:}}
\PY{k+kt}{\texttt{string}} \PY{o}{\texttt{->}}
\PY{o}{\texttt{[}}\PY{n+nl}{\texttt{secret}}\PY{o}{\texttt{]}}
\PY{k+kt}{\texttt{bool}}}. When a dependency is not mentioned at a formal
argument, those on the actual argument are propagated directly to the call
site. So assuming we have some \mbox{\texttt{x} \PY{o}{\texttt{:}}
\PY{o}{\texttt{[}}\PY{n+nl}{\texttt{alice}}\PY{o}{\texttt{]}}
\PY{k+kt}{\texttt{string}}}, the application \texttt{check} \texttt{x} type
checks under \PY{o}{\texttt{[}}\PY{n+nl}{\texttt{secret}}
\PY{n+nl}{\texttt{alice}}\PY{o}{\texttt{]}} \PY{k+kt}{\texttt{bool}}. Similar
reasoning applies to the previously discussed types for integer increment and
addition. Our examples will use this ability going forward. A further complaint
about the type of \texttt{check} regards its impracticality: a useful password
checker must leak its result, so we do not wish the returned
\PY{k+kt}{\texttt{bool}} to be \PY{n+nl}{\texttt{secret}}. This too will be
allayed through \autoref{sec:syntax-typing} and \autoref{sec:examples}, where
we discuss our ability to support \textit{downgrading}\footnote{We use
`downgrading' here to avoid referring individually to declassification or
endorsement, since our deconstruction of information flow will yield a
theoretical divide distinct from the one between confidentiality and
integrity.} using only features already introduced.

The moves we have made so far are analogous to those of
\citet{gouni2025structural}; indeed, part of the preceding exposition has been
adapted from there. However, having exhausted their insights we find ourselves
at an as-yet-incomplete picture. From a theoretical perspective, it is
suspicious that the other aspect of parametricity---information hiding---has
not yet played a role. Practically, we lack the ability to speak about a number
of important secure programming patterns.

For instance, while it is useful to track the dependence of a computation on
password data, we might like to determine that the result of some computation
\textit{cannot} depend on---that is, cannot leak---password data. Our only
chance of doing this at present is to manually verify that its dependency set
does not contain \PY{n+nl}{\texttt{secret}}. But even this falls flat, for
instance, in the presence of polymorphic dependencies \PY{n+nl}{\valpha} which
may later be instantiated to dependency sets containing
\PY{n+nl}{\texttt{secret}}: we would have to verify upon every instantiation that
the resulting set is still free of \PY{n+nl}{\texttt{secret}}. As another example,
we would like to know that data tagged with \PY{n+nl}{\texttt{untrusted}} does
not affect some sensitive computation. We could again look through its
dependency set and manually validate that it is free of
\PY{n+nl}{\texttt{untrusted}}, but this is subject to the same pitfalls as
before. The current state of affairs is evidently unsatisfactory. Conveniently,
filling in the theoretical part of our incomplete picture will turn out to do
the same on the practical side, so let us investigate information hiding.

\subsubsection{Isolating Information Reasoning}
\label{sec:isolating-info-intro}

Having isolated reasoning about flows, can we do the same for information? We
focus in particular on its \textit{hiding}, going by our deconstruction of
parametricity, because this is the essence of data abstraction. To account for
this we introduce new syntax for \textit{anti}-dependencies explicating that on
which a computation \textit{cannot} depend. Specifically, we provide a way to
prohibit some data from being used by certain computations by tagging the
former with anti-dependencies against the latter. Data tagged as such can be
seen as \textit{conditionally} abstract, that is, only from the perspective of
computations possessing those dependencies. This is a generalization of
information reasoning over that afforded by parametricity, as with reasoning
about flows above.

Let us see how this helps us with locally preventing a computation from leaking
secrets, as previously desired. Consider some computation \mbox{\texttt{public}
\PY{o}{\texttt{:}} \PY{o}{\texttt{...}} \PY{o}{\texttt{->}}
\PY{o}{\texttt{[}}\PY{n+nl}{\valpha}\PY{o}{\texttt{]}} \PY{k+kt}{\texttt{int}}}
whose result we would like to declare cannot depend on password information.
With anti-dependencies in hand, this is as simple as inserting one regulating
against \PY{n+nl}{\texttt{secret}} into its return type:
\mbox{\PY{o}{\texttt{...}} \PY{o}{\texttt{->}}
\PY{o}{\texttt{[}}\PY{n+nl}{\valpha}\PY{o}{\texttt{]}}
\PY{o}{\texttt{!}}\PY{o}{\texttt{[}}\PY{n+nl}{\texttt{secret}}\PY{o}{\texttt{]}}
\PY{k+kt}{\texttt{int}}}. If \PY{n+nl}{\valpha} is instantiated to
\PY{n+nl}{\texttt{secret}} when \texttt{public} is used, we will be left with
\mbox{\error{\texttt{[}\texttt{secret}\texttt{]}
\texttt{!}\texttt{[}\texttt{secret}\texttt{]} \texttt{int}} \danger}. This is
an error because a
\PY{o}{\texttt{!}}\PY{o}{\texttt{[}}\PY{n+nl}{\texttt{secret}}\PY{o}{\texttt{]}}
\PY{k+kt}{\texttt{int}} is hidden from computations dependent on
\PY{n+nl}{\texttt{secret}}. \autoref{fig:closed-1} provides further motivation,
protecting sensitive computations from
\PY{o}{\texttt{[}}\PY{n+nl}{\texttt{untrusted}}\PY{o}{\texttt{]}} data.

\begin{figure}
    \begin{BVerbatim}[commandchars=\\\{\}]
\textcolor{lightgray}{1} \PY{k}{let} user_input \PY{o}{:} \PY{o}{[}\PY{n+nl}{untrusted}\PY{o}{]} \PY{k+kt}{string} \PY{o}{=} \PY{l+s}{"saio pae"}
\textcolor{lightgray}{2} \PY{k}{let} neural_net \PY{o}{:} \PY{o}{![}\PY{n+nl}{untrusted}\PY{o}{]} \PY{o}{\texttt{(}}\PY{k+kt}{string} \PY{o}{->} \PY{k+kt}{bool}\PY{o}{\texttt{)}} \PY{o}{=} \PY{o}{...}\phantom{\hspace{5.75em}}
    \end{BVerbatim}
    \vspace{1em}
    \textcolor{lightgray}{\hrule}
    \vspace{1em}
    \begin{BVerbatim}[commandchars=\\\{\}]
\textcolor{lightgray}{3} \PY{k}{let} \PY{o}{_} \PY{o}{:} \PY{o}{[}\PY{n+nl}{secret}\PY{o}{]} \PY{o}{![}\PY{n+nl}{untrusted}\PY{o}{]} \PY{k+kt}{bool} = neural_net password
\textcolor{lightgray}{4} \error{let _ : [untrusted] ![untrusted] bool = neural_net user_input} \danger
    \end{BVerbatim}
    \caption{Preventing Adversarial Input to a Neural Network}
    \label{fig:closed-1}
\end{figure}

We begin by declaring \texttt{user\_input} tagged
\PY{o}{\texttt{[}}\PY{n+nl}{\texttt{untrusted}}\PY{o}{\texttt{]}} and a
function \texttt{neural\_net} whose type is enclosed in
\PY{o}{\texttt{!}}\PY{o}{\texttt{[}}\PY{n+nl}{\texttt{untrusted}}\PY{o}{\texttt{]}}
to indicate that it must not be invoked within an
\PY{n+nl}{\texttt{untrusted}}-dependent computation, because of its inherent
vulnerability to adversarial input. It takes a \PY{k+kt}{\texttt{string}} and
returns a \PY{k+kt}{\texttt{bool}}, using dependency elision to slightly
simplify what ordinarily would have been written
\mbox{\PY{o}{\texttt{[}}\PY{n+nl}{\valpha}\PY{o}{\texttt{]}}
\PY{k+kt}{\texttt{string}} \PY{o}{\texttt{->}}
\PY{o}{\texttt{[}}\PY{n+nl}{\valpha}\PY{o}{\texttt{]}}
\PY{k+kt}{\texttt{bool}}}. We make use of it on line 3, passing it the
\texttt{password} variable from \autoref{fig:password-check} and storing the
result into a \PY{k+kt}{\texttt{bool}} appropriately tagged
\PY{n+nl}{\texttt{secret}} to indicate the dependency. The next line swaps out
\texttt{password} for \texttt{user\_input}, culminating in an error due to the
conflict between the
\PY{o}{\texttt{[}}\PY{n+nl}{\texttt{untrusted}}\PY{o}{\texttt{]}} argument and
the restriction against such a dependency.

Observe our phrasing around the mechanics of anti-dependencies: \textbf{a
conflict between dependencies and anti-dependencies arises when a computation
at some dependency \textit{contains} data which has declared an anti-dependency
on the same.} In terms of the prior example, we care specifically about the
existence of \PY{o}{\texttt{[}}\PY{n+nl}{\texttt{untrusted}}\PY{o}{\texttt{]}}
with
\PY{o}{\texttt{!}}\PY{o}{\texttt{[}}\PY{n+nl}{\texttt{untrusted}}\PY{o}{\texttt{]}}
nested \textit{inside} it. One might wonder whether the other order of
composition
\mbox{\PY{o}{\texttt{!}}\PY{o}{\texttt{[}}\PY{n+nl}{\texttt{untrusted}}\PY{o}{\texttt{]}}
\PY{o}{\texttt{[}}\PY{n+nl}{\texttt{untrusted}}\PY{o}{\texttt{]}}
\PY{k+kt}{\texttt{bool}}} results in an error. It does not, because in this
case the data which has declared an anti-dependency on
\PY{n+nl}{\texttt{untrusted}} is not yet itself dependent on
\PY{n+nl}{\texttt{untrusted}}, not being under such a dependency. The semantics
of dependencies and anti-dependencies given by \textit{polarity} will quell any
unease about the possibility of mis-specification owing to this subtlety,
however. \textit{Suspending}
\PY{o}{\texttt{[}}\PY{n+nl}{\texttt{untrusted}}\PY{o}{\texttt{]}}
\PY{k+kt}{\texttt{bool}} will render the reverse order of composition inert in
\autoref{sec:syntax-typing}. We look now to it and the rest.

\subsection{A Preview of the Rest}

Parametricity has provided a convenient inroads to exposing our approach to
information flow, motivating its deconstruction into \textit{dependencies} and
\textit{anti-dependencies}. It turns out that the type connectives
\mbox{\PY{o}{\texttt{[}}\PY{n+nl}{\valpha} \PY{n+nl}{\vbeta}
\PY{o}{\texttt{...}}\PY{o}{\texttt{]}}} and
\mbox{\PY{o}{\texttt{!}}\PY{o}{\texttt{[}}\PY{n+nl}{\valpha} \PY{n+nl}{\vbeta}
\PY{o}{\texttt{...}}\PY{o}{\texttt{]}}} correspond to the \textit{open} and
\textit{closed} modalities previously discussed by \citet{sterling2021logical}
in the context of parametricity for program modules, alongside others. In this
work we shift our attention from parametricity to non-interference, driving the
design of information flow types using the complementary dependency tracking
behaviors of these modalities. Their semantics deeply inform the distinction
between information and flows.

Perhaps surprisingly, all the core features of our system have been discussed
at this point; what remains is to reveal how. \autoref{sec:background} looks to
the broader body of work on information flow, isolating two distinct varieties
of dependency tracking which appear independently and interchangeably
throughout. In \autoref{sec:syntax-typing} we establish the correspondence of
these varieties to the open and closed modalities, showing that though taken
separately they may seem similar, taken together their differences shine. We
walk through a number of further examples in \autoref{sec:examples} which
exploit the interaction between our modalities to exhibit full-spectrum
information flow reasoning encompassing both confidentiality and integrity,
furthermore reconstructing \textit{robust downgrading}.
\autoref{sec:metatheory} formally grounds these examples, demonstrating
non-interference for our system and shedding light on the relationship between
the joint semantics of our modalities and robustness. \autoref{sec:related}
reviews related work, and we conclude in \autoref{sec:conclusion}. In summary,
our contributions are:

\begin{enumerate}
    \item We review the literature on information flow and find \textbf{two
        distinct flavors of dependency tracking,} deployed largely without
        regard to their differences. Each flavor is realized by a different
        formulation of the primary dependency tracking modality. We tease apart
        the metatheory of each, finding that one deals with flows and the other
        with information.
    \item We \textbf{re-cast each flavor in correspondence to either the open
        or the closed modality,} illuminating their semantics from the
        point of view of proof-theoretic polarity and modal type theory. We
        explicate why one corresponds to \textit{dependencies} and the other to
        \textit{anti-dependencies} when appearing in concert, despite acting
        analogously when apart.

    \item \textbf{We reconstruct full-spectrum information flow, covering both
        confidentiality and integrity, through our modalities and their joint
        interaction.} We push this further to recover an improved form of
        robust information flow reasoning over prior work, including
        \textit{robust downgrading}---all without extensions to our framework
        beyond what has already been introduced. A number of practical examples
        serve to ground these points, moreover demonstrating the
        \textbf{benefits to specification simplicity of our framework.}

    \item We give a logical relations argument for non-interference which
        \textbf{realizes robustness as an ordinary 2-hyperproperty,} simplified
        from prior presentations. We owe this to robustness being
        \textit{built in} to the semantics of our modalities; its recovery
        through their interplay places prior work in this area on firm logical
        footing. Remarkably, \textbf{our non-interference property accounts for
        sophisticated downgrading mechanisms without making semantic
        concessions} to them, providing extraordinarily strong modularity
        guarantees. Our results have been mechanized in the Lean theorem
        prover.
\end{enumerate}

Our work ultimately aims to advance the perspective that practical,
full-spectrum information flow programming can be pursued under full-strength
non-interference, and that this can be negotiated through the duality that
results from our deconstruction of information flow. We begin by tracing the
roots of this duality through the literature.

\section{Background: Two Modalities for Dependency Tracking}
\label{sec:background}

Prior work on dependency tracking has deployed a diversity of modalities for
doing so. Of these we survey those that form \textit{monads}
\cite{moggi1989computational}, arranging them along two distinct lines: one
dealing with flows, the other with information. These can respectively be
re-cast as the \textit{open} and \textit{closed} modalities of
\citet{rijke2020modalities}. For narrative ease, we start with the latter. This
section takes notational liberties as needed for clarity in presenting prior
systems. The type system in \autoref{sec:syntax-typing} will be revealed to
import the dependency tracking machinery introduced in this one.

\subsection{Isolating Information, Formally}
\label{sec:isolating-info}

\begin{figure}
\begin{mathpar}
    \inferrule[Seal]
    {\Gamma \vdash e : A}
    {\Gamma \vdash \sealEx{\phi}{e} : T_\phi(A)}

    \inferrule[Unseal]
    {\Gamma \vdash e_1 : T_\phi(A_1) \\ \Gamma, x : A_1 \vdash e_2 : A_2 \\ A_2 \geq \phi}
    {\Gamma \vdash \unsealEx{e_1}{x}{e_2} : A_2}
\end{mathpar}
\vspace{1em}
\textcolor{lightgray}{\hrule}
\vspace{1em}
\begin{tabular}{ c c c c c }
\PY{k+kt}{\texttt{unit}} &
\error{\texttt{bool}} \danger &
\error{$T_{\alpha_1}(\texttt{bool})$} \danger &
$T_{\PY{n+nl}{\alpha_1}\, \PY{n+nl}{\alpha_2}}(\PY{k+kt}{\texttt{bool}})$ &
\PY{k+kt}{\texttt{unit}} \PY{o}{\texttt{*}} \PY{k+kt}{\texttt{unit}} \\
\error{\texttt{bool}} \error{\texttt{*}} \error{\texttt{unit}} \danger &
$T_{\PY{n+nl}{\alpha_1}\, \PY{n+nl}{\alpha_2}}(\PY{k+kt}{\texttt{bool}})$ \PY{o}{\texttt{*}} \PY{k+kt}{\texttt{unit}} &
\error{\texttt{int}} \danger &
\error{\texttt{list} \texttt{unit}} \danger &
\PY{k+kt}{\texttt{bool}} \PY{o}{\texttt{->}} \PY{k+kt}{\texttt{unit}} \\
\end{tabular}
\caption{DCC Dependency Tracking Connective (Top); Compatible Sealings (Bottom)}
\label{fig:dcc}
\end{figure}

The Dependency Core Calculus (DCC) of \citet{abadi1999} pioneered the use of
monads for dependency tracking. It is organized around a monadic type
connective $T_{\PY{n+nl}{\alpha_1} \ldots{} \PY{n+nl}{\alpha_n}}(A)$ denoting
that $A$ is dependent on $\PY{n+nl}{\alpha_1} \ldots{} \PY{n+nl}{\alpha_n}$.
The rules for this connective are given in \autoref{fig:dcc} (top). The
introduction rule \textsc{Seal} annotates an expression at any type with a set
of dependencies $\phi$, where $\phi$ is of the form $\PY{n+nl}{\alpha_1}
\ldots{} \PY{n+nl}{\alpha_n}$. The elimination rule \textsc{Unseal} unwraps the
dependency annotation from $e_1$, extracting the contained value at the inner
type and passing it to $e_2$. The critical portion is the final premise $A_2
\geq \phi$, which states that \textit{all the information in $A_2$ must be
sealed at $\phi$ or above}. We defer a formal description of this sealing
obligation to the next section; it is easier to illustrate instead by example.
\autoref{fig:dcc} (bottom) shows the validity of different choices of $A$ in $A
\geq \PY{n+nl}{\alpha_1\, \alpha_2}$.

In short, \textsc{Unseal} ensures that any computation using sealed data places
all the information it produces under the same level of sealing. Any part of
the result type which risks exfiltrating previously sealed information is
affected by this requirement. Accordingly, the need to \texttt{seal} is
entirely elided when the result type is one incapable of revealing any
information, such as \PY{k+kt}{\texttt{unit}} or \PY{k+kt}{\texttt{unit}}
\PY{o}{\texttt{*}} \PY{k+kt}{\texttt{unit}} or \PY{k+kt}{\texttt{bool}}
\PY{o}{\texttt{->}} \PY{k+kt}{\texttt{unit}}. Sealing as employed here
crystallizes the idea that \textbf{the DCC dependency tracking monad
fundamentally concerns information, or what can be \textit{gleaned about data},
analyzing that represented by types.} Importantly, information is its only
concern, not any other property of the computation at hand. Indeed, $T_\phi(A)$
can be understood as \textit{internalizing} the judgement-level correspondence
between types and information given by $A \geq \phi$.

The DCC monad was recognized as the \textit{closed modality} from later work in
univalent foundations \cite{rijke2020modalities} by \citet{sterling2022sheaf},
who elegantly reworked its categorical semantics. \textit{Modality} here refers
to an \textit{idempotent monad}, specifically a \textit{semantic} view of one.
Idempotency means that we have $T_\phi(T_\phi(A)) \cong T_\phi(A)$. This allows
modalities to be characterized by their \textit{semantics}, or behavior, rather
than by their syntax. Specifically, we say $A$ is $T_\phi$-modal if $T_\phi(A)
\cong A$. This is what sets the DCC dependency tracking connective apart from
an ordinary monad: \textsc{Bind} does not conclude at $T_\phi(-)$, as the
latter would, but at a type $A$ satisfying $A \geq \phi$. This is a lightweight
way of checking that $A$ is $T_\phi$-modal. In other words, our concern is
whether the resulting type \textit{behaves} in the expected way in terms of
revealing information, rather than being conservatively forced to exhibit a
particular syntax. We defer further discussion to \citet{rijke2020modalities},
who provide a comprehensive overview of modalities in the sense deployed here,
and \citet[\S 4.4]{choudhury2022monadic}, who focuses on the DCC modality and
its relationship to standard monads.

Much later work exploits the DCC approach to dependency tracking
\cite{cecchetti2017nonmalleable, shikuma2008proving, liu2024internalizing,
sterling2022sheaf, choudhury2022monadic, bowman2015noninterference,
rajani2025graded, hirsch2021giving}, or references doing so. However,
despite crediting the DCC---or closed---modality, some such works cut much
closer in their chosen flavor of dependency tracking to its complement, the
\textit{open modality}. This is an easy oversight, for the latter has been
comparatively neglected in the information flow literature; we remedy this now.

\subsection{Isolating Flows, Formally}
\label{sec:isolating-flows}

\begin{figure}
    \begin{mathpar}
        \inferrule[Consume]
        {\Gamma; \phi \cup \phi' \vdash e : A}
        {\Gamma; \phi \vdash \lamEx{\_}{e} : \phi' \Rightarrow A}

        \inferrule[Produce]
        {\Gamma; \phi \vdash e : \phi' \Rightarrow A \\ \phi' \subseteq \phi}
        {\Gamma; \phi \vdash \apEx{e}{\phi'} : A}
    \end{mathpar}
\vspace{1em}
\textcolor{lightgray}{\hrule}
\vspace{1em}
\begin{minipage}{0.67\textwidth}
    \centering
    \begin{mathpar}
        \inferrule*[Right=C]
        {\inferrule*[Right=P]
         {x : \PY{n+nl}{\valpha} \Rightarrow{} \PY{k+kt}{\texttt{int}}, y : \PY{n+nl}{\vbeta} \Rightarrow{} \PY{k+kt}{\texttt{int}}; \PY{n+nl}{\alpha} \vdash x : \PY{n+nl}{\valpha} \Rightarrow{} \PY{k+kt}{\texttt{int}} \\ \PY{n+nl}{\valpha} \subseteq \PY{n+nl}{\valpha}}
         {x : \PY{n+nl}{\valpha} \Rightarrow{} \PY{k+kt}{\texttt{int}}, y : \PY{n+nl}{\vbeta} \Rightarrow{} \PY{k+kt}{\texttt{int}}; \PY{n+nl}{\alpha} \vdash \apEx{x}{\alpha} : \PY{k+kt}{\texttt{int}}}}
        {x : \PY{n+nl}{\valpha} \Rightarrow{} \PY{k+kt}{\texttt{int}}, y : \PY{n+nl}{\vbeta} \Rightarrow{} \PY{k+kt}{\texttt{int}}; \varnothing \vdash \lamEx{\_}{\apEx{x}{\alpha}} : \PY{n+nl}{\valpha} \Rightarrow{} \PY{k+kt}{\texttt{int}}}

        \inferrule*[Right=C]
        {\inferrule*[Right=P \danger]
         {x : \PY{n+nl}{\valpha} \Rightarrow{} \PY{k+kt}{\texttt{int}}, y : \PY{n+nl}{\vbeta} \Rightarrow{} \PY{k+kt}{\texttt{int}}; \PY{n+nl}{\alpha} \vdash y : \PY{n+nl}{\vbeta} \Rightarrow{} \PY{k+kt}{\texttt{int}} \\ \error{\vbeta \not\subseteq \valpha}}
         {x : \PY{n+nl}{\valpha} \Rightarrow{} \PY{k+kt}{\texttt{int}}, y : \PY{n+nl}{\vbeta} \Rightarrow{} \PY{k+kt}{\texttt{int}}; \PY{n+nl}{\alpha} \vdash \apEx{y}{\beta} : \PY{k+kt}{\texttt{int}}}}
        {x : \PY{n+nl}{\valpha} \Rightarrow{} \PY{k+kt}{\texttt{int}}, y : \PY{n+nl}{\vbeta} \Rightarrow{} \PY{k+kt}{\texttt{int}}; \varnothing \vdash \lamEx{\_}{\apEx{y}{\beta}} : \PY{n+nl}{\valpha} \Rightarrow{} \PY{k+kt}{\texttt{int}}}
    \end{mathpar}
\end{minipage}
\begin{minipage}{0.32\textwidth}
    \centering
    \begin{tikzpicture}
\node at (0, 0) {\PY{n+nl}{\valpha} $\Rightarrow$ \PY{k+kt}{\texttt{int}} \PY{o}{\texttt{->}}};
\node at (0, -1) {\PY{n+nl}{\vbeta} $\Rightarrow$ \PY{k+kt}{\texttt{int}} \PY{o}{\texttt{->}}};
\node at (0, -2) {\PY{n+nl}{\valpha} $\Rightarrow$ \PY{k+kt}{\texttt{int}} \phantom{\texttt{->}}};

\node[rotate=90] at (-1.4, -1) {\textcolor{Blue}{Unlocks}};
\node[rotate=90] at (1.3, -1) {\textcolor{Red}{Flows}};

\draw [Blue] (-0.74, -1.85) -- (-0.74, -1.75) -- (-1.375, -1.75) -- (-1.375, -1.65);
\draw[->, dashed, Blue] (-1.2, -0.55) -- (-0.74, -0.55) -- (-0.74, -0.75);
\draw [->, Blue] (-1.375, -0.35) -- (-1.375, 0.325) -- (-0.74, 0.325) -- (-0.74, 0.125);

\draw [Red] (-0.74, -0.2) -- (-0.74, -0.3) -- (1.35, -0.3) -- (1.35, -0.5);
\draw [dashed, dash phase=1pt, Red] (-0.74, -1.2) -- (-0.74, -1.3) -- (1.1, -1.3);
\draw [->, Red] (1.35, -1.5) -- (1.35, -2.4) -- (-0.74, -2.4) -- (-0.74, -2.2);
    \end{tikzpicture}
\end{minipage}
    \caption{\citet{shikuma2008proving} Dependency Tracking Connective (Top); Usage Example (Bottom)}
    \label{fig:notdcc}
\end{figure}
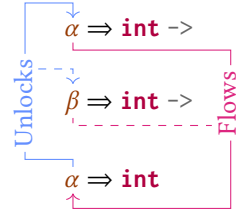

We show the setup for the second kind of modality in \autoref{fig:notdcc}
(top), taken from \citet[\S 2.2]{shikuma2008proving} who introduce it as a
reformulated variant of that from DCC. This modality is still monadic, but
makes no mention of the sealing judgement defining the previous. It works
instead by interacting with the $\phi$ environment newly added to the typing
judgement, denoting the ambient security level. Specifically, it represents the
current set of dependencies of the expression being typed. The introduction
rule \textsc{Consume} takes an expression at type $A$ and ambient dependencies
$\phi \cup \phi'$ and captures the $\phi'$ dependencies within a newly formed
modality $\phi' \Rightarrow A$, moving them into the type. The
elimination rule \textsc{Produce} checks that the ambient security level $\phi$
is higher than the dependency set $\phi'$ of the modality to be eliminated,
then concludes at the inner type.

The notation we use here for the modality is suggestive: it is exactly a
function with a dependency set as its argument type. Through this lens, the
rules for its introduction and elimination mimic exactly those standard for
functions, where the former adds its argument as an assumption to the typing
environment for its body and the latter checks that the current assumptions
suffice for the expected argument. The benefits of viewing this modality as a
function are twofold. First, it permits the dependency tracking mechanics at
play to be understood as a matter of \textit{access token passing},
illuminating its sensitivity to \textit{flows}. Second, it establishes a firm
connection to the \textit{reader monad} for $\phi$, which is precisely the
\textit{open} modality \cite{sterling2021logical}. We review each in turn.

On the first point, consider a function at type
\mbox{\PY{o}{\texttt{(}}\PY{n+nl}{\valpha} $\Rightarrow$
\PY{k+kt}{\texttt{int}}\PY{o}{\texttt{)}} \PY{o}{\texttt{->}}
\PY{o}{\texttt{(}}\PY{n+nl}{\vbeta} $\Rightarrow$
\PY{k+kt}{\texttt{int}}\PY{o}{\texttt{)}} \PY{o}{\texttt{->}}
\PY{n+nl}{\valpha} $\Rightarrow$ \PY{k+kt}{\texttt{int}}}. At first glance, we
might guess that its return value depends only on its first argument. Why?
Stepping through, we assume two modalities at \PY{n+nl}{\valpha} and
\PY{n+nl}{\vbeta} respectively. We conclude with one at \PY{n+nl}{\valpha}.
Reading the modality itself as a function, this means we \textit{assume}
\PY{n+nl}{\valpha} from the overall return type \PY{n+nl}{\valpha}
$\Rightarrow$ \PY{k+kt}{\texttt{int}} and have it available to \textit{unlock}
other data where it is expected as an argument. This is only possible for the
input modality \PY{n+nl}{\valpha} $\Rightarrow$ \PY{k+kt}{\texttt{int}}. For
\PY{n+nl}{\vbeta} $\Rightarrow$ \PY{k+kt}{\texttt{int}} we do not have
\PY{n+nl}{\vbeta} as an assumption, so attempting to apply \textsc{Produce}
fails to unlock its data and extract the contained \PY{k+kt}{\texttt{int}}.
\autoref{fig:notdcc} (bottom) shows candidate function \textit{bodies} for each
case, starting at the bottom with \textsc{Consume} followed by
\textsc{Produce}.

Observe that we are exclusively concerned here with the provenance of various
permissions to unlock data. We pass around such access tokens to speak by proxy
about \textit{flows}; assuming an \PY{n+nl}{\valpha} from the return type and
having it available to use to unlock an argument at \PY{n+nl}{\valpha} permits
that argument to flow to the return value. That is, the presence of an ambient
\PY{n+nl}{\valpha} denotes flows from data guarded by it. \textbf{The nature of
this modality is to indicate \textit{where data is moving} by way of tracking
the authority to access it.} It is defined by \textit{internalizing access
assumptions}, similarly to usual function types internalizing variable typing
assumptions. Unlike the previous modality, it involves no notion of the
information contained in types. One could return an informationally inert
\PY{k+kt}{\texttt{unit}} from some computation, but continue to have to track
which permissions were used to compute it because the flow has already
occurred. Effects bring this special relevance; consider \mbox{\texttt{print}
\PY{o}{\texttt{:}} \PY{n+nl}{\texttt{io}} $\Rightarrow$
\PY{o}{\texttt{(}}\PY{k+kt}{\texttt{string}} \PY{o}{\texttt{->}}
\PY{k+kt}{\texttt{unit}}}\PY{o}{\texttt{)}}. Failing to track
\PY{n+nl}{\texttt{io}} would be tantamount to forgetting potential uses of
\texttt{print}! \autoref{sec:related} will revisit this. For now, we confront
the second point from before.

Given that both its underlying machinery and practical uses differ subtly but
fundamentally from those of the DCC modality, we now have reason to suspect
that the modality here should be cast as something quite distinct from the
latter. It is in reality the \textit{open modality}, introduced by
\citet{rijke2020modalities} as the sibling to the closed modality. The open
modality is defined as a function with its antecedent containing dependencies
$\phi$, a construction familiar to functional programmers as the reader monad
for $\phi$. Besides \citet{shikuma2008proving}, the rules presented here are
also analogous to those for the open modality in
\citet{sterling2021metalanguage}, a fact not observed there or elsewhere. The
open modality makes further appearances in \citet{gouni2025structural} as
discussed in \autoref{sec:intro}, in \citet[\S 3.2]{liu2024internalizing} where
it is again cast as a variant of the DCC monad inspired by
\citet{shikuma2008proving}, and in \citet{sterling2022sheaf} which reviews it
as distinct from the closed modality in an information flow setting. The
lattermost notes the algebraic relationship between the two, but focuses on the
closed modality and its semantic insights towards DCC. In this work we afford
each equal standing, situating both within the same system to take advantage of
the rich information flow reasoning springing forth from their \textit{joint}
semantics.

\autoref{sec:intro} gives one exposition of the divide between flows and
information, in terms of dependencies and anti-dependencies, and we have just
given another here, in terms of differing information flow modalities. We unify
these perspectives in the next section by way of formally introducing our
approach to information flow, which features both modalities. Though both act
to track dependencies when apart, when combined one is flipped into the
anti-dependencies view. Apart from the two throughlines we have seen, there is
a \textit{hybrid} modality within the literature which can be understood as
combining aspects of both. This is briefly summarized for thoroughness in
\autoref{sec:hybrid}, but will not be relevant to our purposes here. We look
now to our type system.

\section{Syntax and Typing}
\label{sec:syntax-typing}

We set up our core language, \textit{parametric information flow}, in the
Call-By-Push-Value framework of \citet{levy2012call} extended with the
just-discussed modalities.\footnote{The Lean mechanization of our type system
uses intrinsically well-typed-and-scoped syntax with binding via De Bruijn
indices, approximated here using explicit scoping directives and
capture-avoiding substitution over named variables.} We use its explicit
treatment of proof-theoretic polarity to further distinguish each modality. We
then discuss their \textit{joint} semantics at a theoretical level, leading
into its practical implications in the section which follows.

\subsection{Information, or Values}
\label{sec:syntax-typing-info}

\begin{figure}
    Dependency Vars $\PY{n+nl}{\alpha_1}, \PY{n+nl}{\alpha_2}, \ldots \in \Delta$ \quad Vars $x_1, x_2, \ldots \in \Gamma$

	\begin{align*}
        \text{Dependencies}~ \phi & \Coloneqq        \epsilon \bnfsep \PY{n+nl}{\alpha}; \phi                             \\
		\text{Value Types}~ A & \Coloneqq            \unitTy \bnfsep \sumTy{A_1}{A_2} \bnfsep \pairTy{A_1}{A_2} \bnfsep \existsTy{\PY{n+nl}{\alpha}}{A} \bnfsep \clTy{\phi}{A} \bnfsep \uTy{\ul{B}} \\
        \text{Values}~ v & \Coloneqq x \bnfsep \unitEx \bnfsep \injlEx{v} \bnfsep \injrEx{v} \bnfsep \pairEx{v_1}{v_2} \bnfsep \packEx{\phi}{v} \bnfsep \sealEx{\phi'}{v} \bnfsep \suspEx{e}
	\end{align*}

    \begin{mathpar}
        \inferrule[T-Var]
        {\\}
        {\typing{\Delta}{\Gamma, x : A}{\phi}{x}{A}}

        \inferrule[T-Unit]
        {\\}
        {\typing{\Delta}{\Gamma}{\phi}{\unitEx}{\unitTy}}

        \inferrule[T-InjL]
        {\typing{\Delta}{\Gamma}{\phi}{v}{A_1} \\ \Delta \vdash A_2}
        {\typing{\Delta}{\Gamma}{\phi}{\injlEx{v}}{\sumTy{A_1}{A_2}}}

        \inferrule[T-InjR]
        {\typing{\Delta}{\Gamma}{\phi}{v}{A_2} \\ \Delta \vdash A_1}
        {\typing{\Delta}{\Gamma}{\phi}{\injrEx{v}}{\sumTy{A_1}{A_2}}}

        \inferrule[T-Pair]
        {\typing{\Delta}{\Gamma}{\phi}{v_1}{A_1} \\ \typing{\Delta}{\Gamma}{\phi}{v_2}{A_2}}
        {\typing{\Delta}{\Gamma}{\phi}{\pairEx{v_1}{v_2}}{\pairTy{A_1}{A_2}}}

        \inferrule[T-Pack]
        {\typing{\Delta}{\Gamma}{\phi}{v}{\Sub{\phi'}{\PY{n+nl}{\alpha}}{A}} \\ \Delta \vdash \phi'}
        {\typing{\Delta}{\Gamma}{\phi}{\packEx{\phi'}{v}}{\existsTy{\PY{n+nl}{\alpha}}{A}}}

        \inferrule[T-Seal]
        {\typing{\Delta}{\Gamma}{\phi}{v}{A} \\ \Delta \vdash \phi'}
        {\typing{\Delta}{\Gamma}{\phi}{\sealEx{\phi'}{v}}{\clTy{\phi'}{A}}}

        \inferrule[T-Susp]
        {\typing{\Delta}{\Gamma}{\phi}{e}{\ul{B}}}
        {\typing{\Delta}{\Gamma}{\phi}{\suspEx{e}}{\uTy{\ul{B}}}}

        \inferrule[T-Sub]
        {\typing{\Delta}{\Gamma}{\phi}{v}{A} \\ \phi'_\open \vdash_{\textcolor{gray}{(\supseteq)}} \phi_\open \\ \Delta \vdash \phi'}
        {\typing{\Delta}{\Gamma}{\phi'}{v}{A'}}

        \inferrule[T-Sub-Type]
        {\typing{\Delta}{\Gamma}{\phi}{v}{A} \\ A \subseteq A' \\ \Delta \vdash A'}
        {\typing{\Delta}{\Gamma}{\phi}{v}{A'}}
    \end{mathpar}
    \caption{Connectives for Values}
    \label{fig:values}
\end{figure}

Call-By-Push-Value bifurcates its connectives between \textit{values} and
\textit{computations}. We start with the values first; their syntax and rules
are shown in \autoref{fig:values}. We have dependency variables
$\PY{n+nl}{\valpha_1}, \PY{n+nl}{\valpha_2}, \ldots$ and term variables $x_1,
x_2, \ldots$, with the former collected into dependency sets $\phi$. The value
types $A$ are all \textit{positive} types in a proof-theoretic sense, with the
slight complication of the \textit{shift} connective $\uTy{\ul{B}}$. $\ul{B}$
will range over the \textit{computation} types, to be discussed in the next
subsection. The form of the value typing judgement is
$\typing{\Delta}{\Gamma}{\phi}{v}{A}$, read as ``under in-scope dependency
variables $\Delta$, in-scope term variables $\Gamma$, and ambient dependencies
$\phi$, the value $v$ has type $A$.''

The rule \textsc{T-Var} for typing variables is standard for the polarized
setting, stating that variables must be at value type. It permits this typing
at any ambient dependencies $\phi$. The introduction rules for \texttt{unit},
sums $\sumTy{A_1}{A_2}$, and positive products $\pairTy{A_1}{A_2}$ likewise
follow largely by convention, threading through the ambient dependencies as
needed. For sums, we check the well-scopedness $\Delta \vdash A$ of the type
not mentioned by the typing assumption in the premise. This is to prevent the
appearance of dependency variables $\PY{n+nl}{\alpha} \not\in \Delta$. For
products, each projection is checked at the same ambient dependencies $\phi$.
This may seem an onerous restriction, but is mediated by a \textit{subsumption}
rule \textsc{T-Sub}, shown in \autoref{fig:values}. The notation $\phi'_\open
\vdash \phi_\open$ can be read as $\phi' \supseteq \phi$; \autoref{sec:versus}
will use it to explain the semantics of our system. This rule allows, for
instance, the lower of two projections with differing security levels to be
heightened as needed. We scope check the heightened dependency set $\phi'$.
$\uTy{\ul{B}}$ is the only \textit{negative} connective among the value types.
\textsc{T-Susp} introduces it by suspending a computation $e$, shifting it to a
value. Its ambient dependencies $\phi$ flow to the conclusion untouched.

\begin{wrapfigure}{r}{0.255\textwidth}
    \centering
    \begin{math}
        \inferrule[Sub-Closed]
        {\phi_\closed \vdash_{\textcolor{gray}{(\subseteq)}} \phi'_\closed \\ A \subseteq A'}
        {\clTy{\phi}{A} \subseteq \clTy{\phi'}{A'}}
    \end{math}
\end{wrapfigure}

The first non-standard connective, fundamentally concerning information flow,
is $\clTy{\phi}{A}$. It is introduced by the rule \textsc{T-Seal}, which we
have already encountered. This connective is a variant of the \textit{closed}
modality, acting analogously to the \textsc{Seal} rule from \autoref{fig:dcc};
its syntax has been updated to follow \citet{sterling2021logical}. Subsumption
mechanics are afforded to it by the rule \textsc{T-Sub-Type} in
\autoref{fig:values}, whose premise $A \subseteq A'$ defines a standard
subtyping relation on types. The case of this relation for the closed modality
is shown to the right, with $\phi_\closed \vdash \phi'_\closed$ read as $\phi
\subseteq \phi'$. \autoref{sec:versus} will illuminate the notation.
\textsc{T-Sub-Type} scope checks the heightened type containing $\phi'$. A
further information flow connective is the existential quantifier
$\existsTy{\PY{n+nl}{\alpha}}{A}$ over dependencies. \textsc{T-Pack} checks the
implementation of an existential at the implementing dependency set $\phi'$
before obscuring $\phi'$ from view in the conclusion, where it is hidden behind
the existential variable $\PY{n+nl}{\valpha}$. This will provide
\textit{downgrading} behavior, as recognized by \citet{gouni2025structural}.
\autoref{sec:examples} will explain by example.

\begin{figure}[h]
    \begin{mathpar}
        \inferrule[S-Unit]
        {\\}
        {\unitTy \geq \phi}

        \inferrule[S-Pair]
        {A_1 \geq \phi \\ A_2 \geq \phi}
        {\pairTy{A_1}{A_2} \geq \phi}

        \inferrule[S-Sealed]
        {\phi_\closed \vdash_{\textcolor{gray}{(\subseteq)}} \phi'_\closed}
        {\clTy{\phi'}{A} \geq \phi}

        \inferrule[S-Unsealed]
        {A \geq \phi}
        {\clTy{\phi'}{A} \geq \phi}

        \inferrule[S-Thunk]
        {\ul{B} \geq \phi}
        {\uTy{\ul{B}} \geq \phi}
    \end{mathpar}
    \caption{Sealing Judgement for Value Types}
    \label{fig:sealing-values}
\end{figure}

Existentials can be approximated through universals \cite[\S
11.3.5]{girard1989proofs}, so prior work \cite{gouni2025structural} does not
treat them explicitly. (Observe that universals have been in use towards
dependency polymorphism since \autoref{sec:intro}, providing the precursors for
downgrading.) In our setting, however, existentials exhibit sufficiently
interesting behavior to invite first-class treatment. To elaborate, recall
\textit{sealing} from \autoref{sec:isolating-info}.
\autoref{fig:sealing-values} defines the sealing judgement for the value types,
which will be the interesting half of its definition. As we saw in
\autoref{fig:dcc}, \unitTy{} is always sealed due to conveying no information;
likewise for pairs with sealed projections. Besides \textsc{S-Unit}, the other
base case for this judgement is provided by \textsc{S-Sealed}, which states
that $\clTy{\phi'}{A} \geq \phi$ is satisfied when $\phi$ is included in
$\phi'$. \textsc{S-Unsealed} gives the case for the closed modality where its
$\phi'$ does not satisfy the inclusion, instead recursing onto the inner type.
\textsc{S-Thunk} recurses onto its contents with the sealing judgement at
computation type, alleged to be straightforward.

So where lies the interesting portion of sealing for values? Exactly in the
cases which have been \textit{omitted}: those for sums and existentials. There
is no case of sealing for sums because a term of type $\sumTy{A_1}{A_2}$ always
has at least two inhabitants, contributed by the left and right injections,
which inevitably reveals information. As for \PY{k+kt}{\texttt{bool}} in
\autoref{fig:dcc}---which can be defined as
$\sumTy{\unitTy}{\unitTy}$---the only way for a sum to exist within a type
satisfying the sealing judgement is to be placed under the closed modality. The
case (or lack of it) for existentials is less clear, but analogous reasoning
applies. Readers may recall that an existential type can be interpreted as a
sum over a family of types; this is exactly the reason for their omission from
sealing. In particular, the implementing dependency set $\phi'$ in
$\packEx{\phi'}{e}$ may differ between different terms of type
$\existsTy{\PY{n+nl}{\valpha}}{A}$. However, one might correctly protest that
there is no way to \textit{use} an existential so as to reveal its
implementation. The preceding behavior is indeed \textit{not} a true semantic
necessity but a concession made for metatheoretic simplicity. This will allow
us to clarify our core results and draw closer comparisons with prior work,
some of which exhibits the same limitation. \autoref{sec:metatheory} will offer
further justification on this point.

Only the value types lack certain cases of sealing; all computation types will
possess routine definitions for the sealing judgement. As embodied by sealing,
information is defined by an analysis of possible inhabitants; it is no
coincidence that this parallels the usual characterization of positive types by
their methods of introduction. This is why $\clTy{\phi}{A}$ exists among the
values: they are the only class of types worth sealing! In other words, only
value types possess even the \textit{potential} to represent information---to
be unable to be sealed---and so the connective for tracking information is
conferred to them. This indicates that \textbf{information is fundamentally a
matter of values.} It is reasonable to wonder, then, what information flow
reasoning gains from computations.

\subsection{Flows, or Computations}
\label{sec:syntax-typing-flows}

\begin{figure}
\begin{align*}
    \text{Computation Types}~ \ul{B} &\Coloneqq \arrTy{A}{\ul{B}} \bnfsep \forallTy{\PY{n+nl}{\alpha}}{\ul{B}} \bnfsep \opTy{\phi}{\ul{B}} \bnfsep \fTy{A} \\
    \text{Expressions}~ e &\Coloneqq
    \lamEx{x}{e} \bnfsep \apEx{e_1}{e_2} \bnfsep \allEx{\PY{n+nl}{\alpha}}{e} \bnfsep \instEx{e}{\phi} \bnfsep \consumeEx{e} \bnfsep \produceEx{e} \bnfsep \retEx{v} \\
        &\hspace{0.5em}\bnfsep \bindEx{e_1}{x}{e_2} \bnfsep \caseEx{v}{x_1}{e_1}{x_2}{e_2} \bnfsep \forceEx{v} \\
        &\hspace{0.5em}\bnfsep \splitEx{v}{x_1}{x_2}{e} \bnfsep \openEx{v}{\PY{n+nl}{\alpha}}{x}{e} \bnfsep \unsealEx{v}{x}{e}
\end{align*}

\begin{mathpar}
    \inferrule[\textcolor{Red}{T-Ret}]
    {\typing{\Delta}{\Gamma}{\phi}{v}{A}}
    {\typing{\Delta}{\Gamma}{\phi}{\retEx{v}}{\fTy{A}}}

    \inferrule[\textcolor{Red}{T-Consume}]
    {\typing{\Delta}{\Gamma}{\phi' \phi}{e}{\ul{B}}}
    {\typing{\Delta}{\Gamma}{\phi'}{\consumeEx{e}}{\opTy{\phi}{\ul{B}}}}

    \inferrule[\textcolor{Red}{T-Produce}]
    {\typing{\Delta}{\Gamma}{\phi'}{e}{\opTy{\phi}{\ul{B}}}}
    {\typing{\Delta}{\Gamma}{\phi\, \phi'}{\produceEx{e}}{\ul{B}}}

    \inferrule[\textcolor{Red}{T-Bind}]
    {\typing{\Delta}{\Gamma}{\phi}{e_1}{\fTy{A}} \\
     \typing{\Delta}{\Gamma, x : A}{\phi}{e_2}{\ul{B}}}
    {\typing{\Delta}{\Gamma}{\phi}{\bindEx{e_1}{x}{e_2}}{\ul{B}}}

    \inferrule[\textcolor{Red}{T-Lam}]
    {\typing{\Delta}{\Gamma, x : A}{\phi}{e}{\ul{B}}}
    {\typing{\Delta}{\Gamma}{\phi}{\lamEx{x}{e}}{\arrTy{A}{\ul{B}}}}

    \inferrule[\textcolor{Blue}{T-Force}]
    {\typing{\Delta}{\Gamma}{\phi}{v}{\uTy{\ul{B}}}}
    {\typing{\Delta}{\Gamma}{\phi}{\forceEx{v}}{\ul{B}}}

    \inferrule[\textcolor{Red}{T-Ap}]
    {\typing{\Delta}{\Gamma}{\phi}{e}{\arrTy{A}{\ul{B}}} \\ \typing{\Delta}{\Gamma}{\phi}{v}{A}}
    {\typing{\Delta}{\Gamma}{\phi}{\apEx{e}{v}}{\ul{B}}}

    \inferrule[\textcolor{Blue}{T-Unseal}]
    {\typing{\Delta}{\Gamma}{\phi}{v}{\clTy{\phi'}{A}} \\
     \typing{\Delta}{\Gamma, x : A}{\phi}{e}{\ul{B}} \\
     \ul{B} \geq \phi'}
    {\typing{\Delta}{\Gamma}{\phi}{\unsealEx{v}{x}{e}}{\ul{B}}}

    \inferrule[\textcolor{Red}{T-All}]
    {\typing{\Delta, \PY{n+nl}{\alpha}}{\Gamma}{\phi}{e}{\ul{B}} \\ \Delta \vdash \phi}
    {\typing{\Delta}{\Gamma}{\phi}{\allEx{\PY{n+nl}{\alpha}}{e}}{\forallTy{\PY{n+nl}{\alpha}}{\ul{B}}}}

    \inferrule[\textcolor{Blue}{T-Split}]
    {\typing{\Delta}{\Gamma}{\phi}{v}{A_1 \otimes A_2} \\
     \typing{\Delta}{\Gamma, x_1 : A_1, x_2 : A_2}{\phi}{e}{\ul{B}}}
    {\typing{\Delta}{\Gamma}{\phi}{\splitEx{v}{x_1}{x_2}{e}}{\ul{B}}}

    \inferrule[\textcolor{Red}{T-Inst}]
    {\typing{\Delta}{\Gamma}{\phi}{e}{\forallTy{\PY{n+nl}{\alpha}}{\ul{B}}} \\
     \Delta \vdash \phi'}
    {\typing{\Delta}{\Gamma}{\phi}{\instEx{e}{\phi'}}{\Sub{\phi'}{\PY{n+nl}{\alpha}}{\ul{B}}}}

    \inferrule[\textcolor{Blue}{T-Case}]
    {\typing{\Delta}{\Gamma}{\phi}{v}{\sumTy{A_1}{A_2}} \\
     \typing{\Delta}{\Gamma, x_1 : A_1}{\phi}{e_1}{\ul{B}} \\
     \typing{\Delta}{\Gamma, x_2 : A_2}{\phi}{e_2}{\ul{B}}}
    {\typing{\Delta}{\Gamma}{\phi}{\caseEx{v}{x_1}{e_1}{x_2}{e_2}}{\ul{B}}}

    \inferrule[\textcolor{Blue}{T-Open}]
    {\typing{\Delta}{\Gamma}{\phi}{v}{\existsTy{\PY{n+nl}{\alpha}}{A}} \\
     \typing{\Delta, \PY{n+nl}{\alpha}}{\Gamma, x : A}{\phi}{e}{\ul{B}} \\
     \Delta \vdash \underline{B} \\
     \Delta \vdash \phi}
    {\typing{\Delta}{\Gamma}{\phi}{\openEx{v}{\PY{n+nl}{\alpha}}{x}{e}}{\ul{B}}}
\end{mathpar}
\caption{Connectives for Computations}
\label{fig:computations}
\end{figure}

We show the rules for computation types $\ul{B}$ in \autoref{fig:computations}.
These comprise the \textit{negative} connectives---save again for the shift
connective $\fTy{A}$ which forms the sibling to $\uTy{\ul{B}}$ and is positive,
inserting values into the computation domain. They follow the judgement
$\typing{\Delta}{\Gamma}{\phi}{e}{\ul{B}}$, read as ``under in-scope dependency
variables $\Delta$, in-scope term variables $\Gamma$, and ambient dependencies
$\phi$, the expression $e$ has type $\ul{B}$.'' The new rules for computation
connectives live alongside the elimination rules for value types in the
computation domain because the latter also eliminate into motives at
computation type. The former set of rules are highlighted in
\textcolor{Red}{\textbf{red}} and the latter in
\textcolor{Blue}{\textbf{blue}}.

Starting with function types $\arrTy{A}{\ul{B}}$, we first show the relevant
case of sealing below. The argument type $A$ is irrelevant as far as
sealing is concerned because one can ultimately only leak information through
the return type $\ul{B}$. So it recurses straightforwardly onto this inner
type; the
\begin{wrapfigure}{r}{0.12\textwidth}
    \centering
    \begin{math}
        \inferrule[S-Arrow]
        {\ul{B} \geq \phi}
        {\arrTy{A}{\ul{B}} \geq \phi}
    \end{math}
\end{wrapfigure}
remaining cases of sealing for computation types proceed analogously
and straightforwardly, as promised. The introduction and elimination rules for
function types---\textsc{T-Lam} and \textsc{T-Ap}---likewise present no
surprises. Note that in the case of \textsc{T-Ap}, the function expression $e$
and its argument $v$ must be at the same ambient dependencies $\phi$.
Discrepancies between them are, as previously discussed for products,
negotiated via a subsumption rule analogous to \textsc{T-Sub} in
\autoref{fig:values}, operating instead on the computation judgement. The
situation is identical for all future rules concerning multiple typed values or
expressions, so we will make no further note of it.

The rule \textsc{T-All} introduces universal quantification over dependencies
$\forallTy{\PY{n+nl}{\alpha}}{\ul{B}}$, binding the \PY{n+nl}{\valpha} in scope
for the typing premise under a universal quantifier and removing it from the
dependency variable environment in the conclusion as it does so. The auxiliary
premise \mbox{$\Delta \vdash \phi$} ensures the ambient dependency set $\phi$
of the typing premise does not mention the just-bound dependency
\PY{n+nl}{\valpha}, preventing \PY{n+nl}{\valpha} from appearing in the
concluding ambient dependencies where it is no longer in scope. The elimination
rule \textsc{T-Inst} substitutes the instantiating dependency set $\phi'$ for
the bound dependency variable \PY{n+nl}{\valpha} within the inner type $\ul{B}$
of the quantifier. $\phi'$ is checked for well-scopedness.

\begin{wrapfigure}{r}{0.25\textwidth}
    \centering
    \begin{math}
        \inferrule[Sub-Open]
        {\phi'_\open \vdash_{\textcolor{gray}{(\supseteq)}} \phi_\open \\ \ul{B} \subseteq \ul{B}'}
        {\opTy{\phi}{\ul{B}} \subseteq \opTy{\phi'}{\ul{B}'}}
    \end{math}
\end{wrapfigure}
The rules \textsc{T-Consume} and \textsc{T-Produce} for introducing and
eliminating $\opTy{\phi}{B}$ look familiar: they are the rules from
\autoref{fig:notdcc} for our variant of the open modality! As in the case of
the closed modality we have updated its syntax, replacing the function notation
within the connective \cite{sterling2021logical} and its terms
\cite{gouni2025structural} in favor of that used by prior work, but the
structure of both rules remains intact. Also as for the closed modality it
enjoys subsumption, given by a rule analogous to
\textsc{T-Sub-Type} from \autoref{fig:values} but operating instead on
computation types. We give the relevant case here; it allows its indexing
dependency set $\phi$ to be raised to any superset $\phi'$, so is structurally
identical to that for the closed modality. We use different notation this time
to reflect the contravariance owed to the open modality due to its reading as a
function from \autoref{sec:isolating-flows}, again to be explained shortly in
\autoref{sec:versus}. Comparing \textsc{Produce} and \textsc{T-Produce}, rather
than checking that the ambient dependencies in the conclusion contain the
dependencies guarding the modality, as in the former, the latter releases the
dependencies $\phi$ guarding the modality into the concluding ambient set. In
both cases, the modality may only be unlocked in the presence of a sufficiently
high concluding ambient dependency set, guaranteeing its use is tracked. The
symmetry of our setup for the open modality exposes its proof-theoretic
harmony---introduction followed by elimination is type-preserving.

Recall that the ambient dependencies denote the flows in play, under the view
that they represent permissions to unlock and use data---guarded by other open
modalities---within the current computation. The only connective whose rules
directly affect the ambient dependencies is the open modality itself, which
\textit{lives entirely within the realm of computations.} Of these rules,
observe that \textsc{T-Produce} uses the ambient dependencies---or
permissions---to extract data. Its invocation can be seen as the definitive
point where the flow occurs. It is also the eliminator for the open modality,
which being a negative type is exactly characterized by it; the open modality
is \textit{defined by causing flows.} This indicates that \textbf{flows are
fundamentally a matter of computations,} which should come as no surprise given
their inherent connection to the dynamic behavior of programs.

A couple complications threaten this perspective. First, the subsumption rule
\textsc{T-Sub} from \autoref{fig:values} also modifies the ambient
dependencies. However, this modification is vacuous, raising the ambient
security level as needed for compatibility between subexpressions. We will lend formal
support to the idea that the ambient security level only \textit{necessarily}
changes for computations in \autoref{sec:metatheory}, where we show subsumption
to be admissible under a semantics which behaves as such. Second, why are
ambient dependencies present in the value judgement at all? The permissions
they represent cannot be applied within the value judgement to realize flows.
These are indeed rather \textit{flows-to-be}, enabling flows when the value
becomes part of a computation. \autoref{sec:simultaneous} will hint, and
\autoref{sec:metatheory} will formalize, that knowing about flows-to-be at the
value level bridges our two modalities. Only ambient dependencies within the
computation judgement denote active flows.

The remaining rules closely follow their standard formulations, threading
through ambient dependencies as needed. Of note, \textsc{T-Open} eliminates
existential quantifiers, binding within $e$ the existential dependency variable
\PY{n+nl}{\valpha} and a variable at type $A$ mentioning it. The type $\ul{B}$
and ambient dependencies $\phi$ of this expression are checked to be
well-scoped under a dependency variable environment not containing
\PY{n+nl}{\valpha}. This means one must eventually downgrade or otherwise avoid
mentioning the existentially quantified dependency variable in the result of
any computation using it. \autoref{sec:examples} will exploit this restriction
towards practical programming.

As promised, our type system features both the open and closed variants of
dependency tracking. We have furthered here the characterization of information
versus flows initiated in the previous section, aligning each flavor of
dependency tracking with the value-computation duality of Call-By-Push-Value.
Following \citet{levy2012call}, \textit{a value is whereas a computation does}.
While the former views programs as inert data, the latter is intertwined with
their dynamic behavior. The insights in this section give rise to the parallel
slogan \textbf{\textit{information is whereas a flow does}; the distinction
between the two mirrors the value-computation duality.} Having said this, we
have thus far left by the wayside the characterization of information versus
flows via \textit{dependencies} and \textit{anti-dependencies} from the
introduction. We remedy this now, unifying it with that just discussed.

\subsection{Information and Flows Simultaneously}
\label{sec:simultaneous}

We have just shown that the open and closed modalities respectively capture
reasoning regarding information versus flows. However, \autoref{sec:intro}
teased apart these same ideas via \textit{dependencies} and
\textit{anti-dependencies}. The \textit{joint} semantics of the open and closed
modalities give rise to a consistent resolution of this discrepancy.
Afterwards, we comment on the distinction between the standard formulation of
the open and closed modalities and that given by the preceding exposition.

We hinted at the nature of the reconciliation between these differing views
earlier, with \autoref{sec:intro} stating that \citet{gouni2025structural}
exploit the \textit{flows} or \textit{dependencies} fragment, and
\autoref{sec:isolating-flows} relating their information flow connective to the
open modality. So \textbf{the open modality is associated to reasoning about
dependencies,} with the remaining possibility being that \textbf{the closed
modality is linked to anti-dependencies.} In \autoref{sec:intro} the former was
written \PY{o}{\texttt{[}}\PY{n+nl}{\valpha} \PY{n+nl}{\vbeta}
\PY{o}{\ldots}\PY{o}{\texttt{]}} and the latter
\mbox{\PY{o}{\texttt{!}}\PY{o}{\texttt{[}}\PY{n+nl}{\valpha} \PY{n+nl}{\vbeta}
\PY{o}{\ldots}\PY{o}{\texttt{]}}}. But why this grouping---of flows,
dependencies, and the open modality on the one side, and information,
anti-dependencies, and the closed modality on the other?

\begin{wrapfigure}{r}{0.22\textwidth}
    \centering
    \begin{math}
        \inferrule
        {A \geq \phi' \\ \phi_\open \vdash \phi'_\closed}
        {\Delta\; \Gamma\; \phi \vdash v_1 \equiv v_2 : A}
    \end{math}
\end{wrapfigure}

The answer lies in the interaction between our modalities. In particular, they
satisfy (roughly) the equation to the right, adapted from
\citet{sterling2021metalanguage}, which states that all terms sealed at $\phi'$
are equivalent under ambient dependencies $\phi$. The effect of the premise
\mbox{$\phi_\open \vdash \phi'_\closed$}---which we will fully elucidate
next---is to check that $\phi$ and $\phi'$ have \textit{overlapping elements}.
In other words, their intersection is non-empty. Recall from
\autoref{sec:isolating-info} that $\phi'$ in $A \geq \phi'$ refers to
dependencies situated within the closed modality, which internalizes this
judgement, and from \autoref{sec:isolating-flows} that the ambient dependencies
$\phi$ can be thought of as connected to the open modality, being internalized
by it. The meaning of the equation, then, is that data sealed within the closed
modality at some set of dependencies is rendered \textit{indistinguishable}
under the purview of an ambient set of dependencies which overlaps with it. The
dependencies within the closed modality therefore regulate against their
mention in the ambient dependencies, or the open modality; values at the former
are \textit{anti-dependent} on computations at the latter. We will continue to
refer to $\PY{n+nl}{\valpha}$s as dependencies, identifying anti-dependencies
by placement within the closed modality. Indistinguishability corresponds to
the central soundness property, amounting to faithfully tracking both open and
closed: \textit{data at non-overlapping type cannot depend on data at
overlapping type}. \autoref{sec:metatheory} will need to enrich the preceding
to account for downgrading.

In summary, the action of the closed modality is to protect data from being
accessed in settings where one has assumed permissions, and therefore indicated
flows, adverse to it. \textbf{While open represents flows, closed hides
information from them.} Note that overlapping, or running into a
\textit{conflict} in the nomenclature of \autoref{sec:isolating-info-intro},
does not necessarily lead to a type error. The stated property is merely one of
separation, not of well-formedness: it is perfectly valid to have a value at
type $\clTy{\phi}{A}$ with ambient dependencies $\phi$. It simply cannot be
depended upon by any non-conflicted parts of the program. A practical
implementation invites producing type errors, as implied by the examples in
\autoref{fig:closed-1}, but how this might be done is not prescribed
theoretically. We bring this section to a close by briefly discussing the
issues of composition order and the relationship of our modalities to the
standard ones.

\subsubsection{On Composition Order}
\label{sec:order}

What we have just discussed corresponds to the order of composition
\mbox{\PY{o}{\texttt{[}}\PY{n+nl}{\valpha} \PY{n+nl}{\vbeta}
\PY{o}{\ldots}\PY{o}{\texttt{]}}
\PY{o}{\texttt{!}}\PY{o}{\texttt{[}}\PY{n+nl}{\valpha} \PY{n+nl}{\vbeta}
\PY{o}{\ldots}\PY{o}{\texttt{]}} \PY{k+kt}{\texttt{t}}} seen in
\autoref{sec:intro}, with anti-dependencies at the closed modality situated
within dependencies at the open modality. What about the reverse order? This
turns out to be inert. To begin, view the examples of the introduction under a
standard call-by-value semantics. The translation from call-by-value into the
explicitly polarized core language presented in this section proceeds routinely
via the procedure given by \citet[\S 2.7.1]{levy2012call}.

In particular, setting $\phi = \PY{n+nl}{\alpha}; \PY{n+nl}{\beta};
\PY{o}{\text{\ldots}}$, we have
$\llangle\PY{o}{\texttt{[}}\PY{n+nl}{\valpha}\; \PY{n+nl}{\vbeta}\;
\PY{o}{\text{\ldots}}\PY{o}{\texttt{]}}\; \PY{k+kt}{\texttt{t}}\rrangle
\rightsquigarrow
\uTy{\opTy{\phi}{\fTy{\llangle\PY{k+kt}{\texttt{t}}\rrangle}}}$ for the
translation of the open modality, and
$\llangle\PY{o}{\texttt{![}}\PY{n+nl}{\valpha}\; \PY{n+nl}{\vbeta}\;
\PY{o}{\text{\ldots}}\PY{o}{\texttt{]}}\; \PY{k+kt}{\texttt{t}}\rrangle
\rightsquigarrow \clTy{\phi}{\llangle\PY{k+kt}{\texttt{t}}\rrangle}$ for
the closed. The call-by-value open modality
\PY{o}{\texttt{[}}\PY{n+nl}{\valpha} \PY{n+nl}{\vbeta}
\PY{o}{\ldots}\PY{o}{\texttt{]}} is surrounded by the shift connectives
\textsf{U} and \textsf{F} due to being a negative type, and the call-by-value
closed modality \mbox{\PY{o}{\texttt{!}}\PY{o}{\texttt{[}}\PY{n+nl}{\valpha}
\PY{n+nl}{\vbeta} \PY{o}{\ldots}\PY{o}{\texttt{]}}} is left untouched because
it is a positive type. Translations for all other types are given analogously.
The translation for the reverse composition is then
\mbox{$\llangle\PY{o}{\texttt{!}}\PY{o}{\texttt{[}}\PY{n+nl}{\valpha}\;
\PY{n+nl}{\vbeta}\; \PY{o}{\text{\ldots}}\PY{o}{\texttt{]}}\;
\PY{o}{\texttt{[}}\PY{n+nl}{\valpha}\; \PY{n+nl}{\vbeta}\;
\PY{o}{\text{\ldots}}\PY{o}{\texttt{]}}\; \PY{k+kt}{\texttt{t}}\rrangle
\rightsquigarrow \clTy{\phi}{ \llangle\PY{o}{\texttt{[}}\PY{n+nl}{\valpha}\;
\PY{n+nl}{\vbeta}\; \PY{o}{\text{\ldots}}\PY{o}{\texttt{]}}\;
\PY{k+kt}{\texttt{t}}\rrangle} \rightsquigarrow
\clTy{\phi}{\uTy{\opTy{\phi}{\fTy{\llangle \PY{k+kt}{\texttt{t}}
\rrangle}}}}$}.

Reifying the reverse composition in Call-By-Push-Value allows us to see why it
poses no danger, even if someone mistakenly assumed it would be considered a
conflict and therefore covered by the soundness guarantee. The closed modality
at $\phi$ does not \textit{directly} contain the open modality at $\phi$ at the
end of the translation, but rather encapsulates it \textit{inside a suspension}
\textsf{U}. Accordingly, these open dependencies $\phi$ and the flows they
represent are not ambient but suspended with respect to the closed modality, so
are not considered conflicting! It is only once the inner computation runs and
exposes its dependencies to the closed modality that any danger arises of
information protected under anti-dependencies participating in the flows
denoted by conflicting ambient dependencies.

\subsubsection{Comparing to the Standard Modalities}
\label{sec:versus}

In certain cases, we have been careful to refer to our connectives
$\clTy{\phi}{A}$ and $\opTy{\phi}{\ul{B}}$ as \textit{variants} of the usual
closed and open modalities. Those in our setting have been modified to increase
specification readability and confer modular downgrading behavior. The
formation rules for the standard modalities might be written as on the left of
\autoref{fig:pairing}.

\begin{figure}
\centering
\begin{minipage}{0.53\textwidth}
\begin{mathpar}
    \inferrule
    {P : \mathsf{Prop} \\ A : \mathsf{Type}_\mathbb{V}}
    {\clTy{P}{A} : \mathsf{Type}_\mathbb{V}}

    \inferrule
    {P : \mathsf{Prop} \\ \ul{B} : \mathsf{Type}_\mathbb{C}}
    {\opTy{P}{\ul{B}} : \mathsf{Type}_\mathbb{C}}

\end{mathpar}
\end{minipage}
\hfill
\textcolor{lightgray}{\vrule}
\hfill
\begin{minipage}{0.41\textwidth}
    \begin{align*}
    \llbracket \PY{n+nl}{\alpha} \rrbracket &\triangleq p~ \text{where}~ p = \mathcal{P}[\PY{n+nl}{\alpha}] \\
    \llbracket \epsilon \rrbracket &\triangleq (\top, \bot) \\
    \llbracket \PY{n+nl}{\alpha}; \phi \rrbracket &\triangleq (\llbracket\PY{n+nl}{\alpha}\rrbracket \wedge \llbracket \phi \rrbracket.\pi_1, \llbracket\PY{n+nl}{\alpha}\rrbracket \vee \llbracket \phi \rrbracket.\pi_2)
    \end{align*}
\end{minipage}
\caption{Signatures of the Standard Modalities (Left); Interpretation of Dependencies (Right)}
\label{fig:pairing}
\end{figure}

Each takes a proposition $P$ followed by a value or computation type $A$ or
$\ul{B}$ and returns $\clTy{P}{A}$ or $\opTy{P}{\ul{B}}$. Working with
\textit{propositions} rather than sets of dependencies is the essence of the
difference. The meaning of $\opTy{P}{\ul{B}}$, recalling its interpretation as
a function, is to say that we have $\ul{B}$ under assumption $P$. And the
meaning of $\clTy{P}{A}$ is that when $P$ holds its information becomes
indistinguishable, or $\clTy{P}{A} \cong \unitTy$. For nomenclature, when
$\opTy{P}{\ul{B}} \cong \unitTy$, we say $\ul{B}$ is $\open_P$-connected.
\autoref{fig:pairing} shows to the right the propositional structure of a
$\phi$ to help clarify our connection to this setup. Assume a logic with the
usual truth values $\top$ and $\bot$, in addition to a set of \textit{included
middles} $\mathcal{P}$---extra truth values neither $\top$ nor $\bot$. We
generate one such value $p \in \mathcal{P}$ to interpret each in-scope
dependency $\PY{n+nl}{\alpha}$. We then interpret dependency sets $\phi$ as a
pair where the first projection is a sequence of conjunctions over its elements
and the second a sequence of disjunctions. As an example, where
$\llbracket\PY{n+nl}{\alpha}\rrbracket \triangleq p_1$ and
$\llbracket\PY{n+nl}{\beta}\rrbracket \triangleq p_2$, we have
$\llbracket\PY{n+nl}{\alpha}; \PY{n+nl}{\beta}; \epsilon\rrbracket \triangleq
(p_1 \land p_2, p_1 \lor p_2)$.

We now define our modalities as the standard ones precomposed with the
projection maps: $\open \circ \pi_1$ and $\closed \circ \pi_2$. The first
premise of their formation rules changes to take a \textit{pair} of
propositions $\mathsf{Prop} \times \mathsf{Prop}$ representing a $\phi$. This
means $\opTy{\phi}{\ul{B}}$ is effectively indexed by a conjunction, and
$\clTy{\phi}{A}$ by a disjunction. \textbf{The notations $\phi_\open$ and
$\phi_\closed$ view the dependencies in $\phi$ as logical atoms and conjunct or
disjunct them together,} corresponding to taking either the first or second
projections. Their entailment is then as usual. We choose this rather than a
more straightforward presentation of each $\phi$ as a set because it highlights
a number of desirable properties. \textsc{T-Sub} and the case of subsumption
for $\opTy{\phi}{\ul{B}}$ are contravariant with respect to $\phi_\open$, as
should be expected of judgemental assumptions and function argument types,
while $\clTy{\phi}{A}$ is covariant for $\phi_\closed$. Further, the
indistinguishability condition $\phi_\open \vdash \phi_\closed$ encodes the
desired overlapping property by checking that a string of conjunctions entails
one of disjunctions, while mirroring the standard one stating that the ambient
assumptions entail a closed proposition. \textbf{We respectively refer to
$\opTy{\phi}{\ul{B}}$ and $\clTy{\phi}{A}$ as the \textit{quasi-open} and
\textit{quasi-closed}\footnote{No relevance to the modality of the same name of
\citet{schultz2017temporal}.} modalities to reflect their non-standard
formulation,} particularly when their divergence from the standard modalities
is relevant.

Specification simplicity ranks among the practical advantages gained from this
divergence. With the standard modalities we have $\opTy{P}{\opTy{Q}{\ul{B}}}
\cong \opTy{P \land Q}{\ul{B}}$ by currying and $\clTy{P}{\clTy{Q}{A}} \cong
\clTy{P \lor Q}{A}$ straightforwardly from \citet[Ex. 1.8]{rijke2020modalities}
and \citet[Lem. 8.5.9]{hottbook}. As such, we would need to support both
conjunction and disjunction operators within our information flow
specifications, precluding the flat sets $\phi$ we currently use. This
semantics furthermore turns out to grant downgrading behavior, because the base
case for $\epsilon$ in \autoref{fig:pairing} evades the indistinguishability
condition; \autoref{sec:metatheory} will explain. Finally, note that with
only two truth values $\top$ and $\bot$, we have $\opTy{\top}{\ul{B}} \cong
\ul{B}$ and $\clTy{\top}{A} \cong \unitTy$, and $\opTy{\bot}{\ul{B}} \cong
\unitTy$ and $\clTy{\bot}{A} \cong A$. The included middles allow us to mention
the same dependency \PY{n+nl}{\valpha} in both modalities without trivializing
either. This is a technique used by prior work \cite[\S
2.1.1]{grodin2026abstraction}, originating from the anti-classical foundations
of the modalities. \autoref{sec:metatheory} will show the metatheory to be
organized around these extra truth values, and \autoref{sec:related} will
discuss further formal consequences of the preceding designs. We consider now
their programming consequences.

\section{Examples}
\label{sec:examples}

Theoretical concerns have received much attention; we now seek to reap the
practical benefits of the preceding intuitions surrounding duality in
information flow. The closed modality will be shown to provide
\textit{robustness} reasoning, leading to a unifying view on \textit{robust
declassification} and \textit{transparent endorsement} from the literature. We
finish by showing how to program with conflicts \textit{modularly}. In order to
avoid syntactic clutter from the $\uTy{-}$ and $\fTy{-}$ shifts of the
polarized setting, we work here in a call-by-value language derived from it;
this follows the translation from \autoref{sec:order}. When we use monospace
type notation as in \PY{k+kt}{\texttt{t}} we refer to call-by-value types.

\subsection{Robustness and the Closed Modality}

The closed modality can be seen as fundamentally providing reasoning about
\textit{robustness} properties. The original definition of robustness, as
explored by later work \cite{gordon2001authenticity, bugliesi2011resource,
swasey2017robust, mackay2022necessity, zdancewic2001robust}, is as follows.

\begin{quote}
    The degree to which a system or component can function correctly in the
    presence of invalid or stressful environmental conditions.
    \hfill \citet{ieee1990standard}
\end{quote}

An integer under the closed modality as in
\PY{o}{\texttt{!}}\PY{o}{\texttt{[}}\PY{n+nl}{\texttt{untrusted}}\PY{o}{\texttt{]}}
\PY{k+kt}{\texttt{int}} can be understood as stating that it will render itself
inaccessible within an environment where it is unable to behave according to
its intended semantics. Namely, one where computations involving data at
\PY{o}{\texttt{[}}\PY{n+nl}{\texttt{untrusted}}\PY{o}{\texttt{]}} are in play.
More broadly, the closed modality defines the permissible environments for the
data it contains. The cumulative shape of these environmental restrictions
gives the robustness properties of the program, explicating the contexts in
which its various fragments may be used. We wish to advance here the
perspective that \textbf{anti-dependencies speak about robustness against
dependencies.}

\begin{figure}
\begin{BVerbatim}[commandchars=\\\{\}]
\textcolor{lightgray}{1} \PY{k}{let} password \PY{o}{:} \PY{o}{[}\PY{n+nl}{secret}\PY{o}{]} \PY{k+kt}{string} \PY{o}{=} \PY{o}{\{} \PY{l+s}{"takver"} \PY{o}{\}}
\textcolor{lightgray}{2} \PY{k}{let} user_input \PY{o}{:} \PY{o}{[}\PY{n+nl}{untrusted}\PY{o}{]} \PY{k+kt}{string} \PY{o}{=} \PY{o}{\{} \PY{l+s}{"saio pae"} \PY{o}{\}}
\textcolor{lightgray}{3} \PY{k}{let} neural_net \PY{o}{:} \PY{o}{![}\PY{n+nl}{untrusted}\PY{o}{]} \PY{o}{\texttt{(}}\PY{k+kt}{string} \PY{o}{->} \PY{k+kt}{bool}\PY{o}{\texttt{)}} \PY{o}{=} \PY{k}{seal} \PY{o}{...}
\end{BVerbatim}
\vspace{1em}
\textcolor{lightgray}{\hrule}
\vspace{1em}
\begin{BVerbatim}[commandchars=\\\{\}]
\textcolor{lightgray}{4} \error{let _ : [untrusted] ![untrusted] bool =}
\textcolor{lightgray}{5}     \error{\{ unseal neural_net as f in seal (f user_input) \}} \danger
\textcolor{lightgray}{6} \PY{k}{let} \PY{o}{_} \PY{o}{:} \PY{o}{[}\PY{n+nl}{secret}\PY{o}{]} \PY{o}{![}\PY{n+nl}{untrusted}\PY{o}{]} \PY{k+kt}{bool} \PY{o}{=}
\textcolor{lightgray}{7}     \PY{o}{\{} \PY{k}{unseal} neural_net \PY{k}{as} f \PY{k}{in} \PY{k}{seal} \PY{o}{(}f password\PY{o}{)} \PY{o}{\}}\phantom{\hspace{3.25em}}
\end{BVerbatim}
    \caption{Revisiting Adversarial Input with Fully Explicit Syntax (see \S 4.1.1 for notation)}
\label{fig:revisiting}
\end{figure}

\subsubsection{Towards Integrity}

\autoref{fig:closed-1} already reviewed a case of such robustness restrictions
towards integrity by excluding the \texttt{neural\_net} function, vulnerable to
adversarial inputs, from
\PY{o}{\texttt{[}}\PY{n+nl}{\texttt{untrusted}}\PY{o}{\texttt{]}} environments.
\autoref{fig:revisiting} revisits this example with the hindsight of the full
formalism behind us, transforming it to use more explicit syntax to expose the
underlying constructs used. The first line now wraps its string
\PY{l+s}{\texttt{"takver"}} in \PY{o}{\texttt{\{}} \PY{o}{\texttt{...}}
\PY{o}{\texttt{\}}}, which provides notation for the $\suspEx{\consumeEx{-}}$
corresponding to the open modality in its type; we do the same on the next. The
definition of \texttt{neural\_net} on the third line now begins with a
\PY{k}{\texttt{seal}} corresponding to the occurrence of the closed modality in
its type, presumably followed by a function implementation.

Moving to the usages of these variables in the bottom half, as before we have
an expression at
\mbox{\PY{o}{\texttt{[}}\PY{n+nl}{\texttt{secret}}\PY{o}{\texttt{]}}
\PY{o}{\texttt{![}}\PY{n+nl}{\texttt{untrusted}}\PY{o}{\texttt{]}}
\PY{k+kt}{\texttt{bool}}} on line 6. However, its implementation looks rather
more complex. Stepping into the \mbox{\PY{o}{\texttt{\{}} \PY{o}{\texttt{...}}
\PY{o}{\texttt{\}}}} corresponding to
\PY{o}{\texttt{[}}\PY{n+nl}{\texttt{secret}}\PY{o}{\texttt{]}}, we first
\PY{k}{\texttt{unseal}} \texttt{neural\_net} to extract the inner component of the
\PY{k}{\texttt{seal}}, binding it as \texttt{f}. We then call this function on
the same argument as before,
re\PY{k}{\texttt{seal}}ing the result. Note that there is no syntax deployed
for invoking $\produceEx{\forceEx{-}}$ on \texttt{password}. When data at the
open modality appears within \mbox{\PY{o}{\texttt{\{}} \PY{o}{\texttt{...}}
\PY{o}{\texttt{\}}}}, we automatically do so. The \textit{idiom brackets} of
\citet{mcbride2008applicative} provide the inspiration for this mechanic,
exploiting the visual boundary given by the brackets to capture all usages of
the open modality within. The erroneous usage on line 4 once again has the same
type as before, and its implementation follows analogously. After unsealing
\texttt{neural\_net} at
\PY{o}{\texttt{![}}\PY{n+nl}{\texttt{untrusted}}\PY{o}{\texttt{]}} we call it
on \texttt{user\_input} at
\PY{o}{\texttt{[}}\PY{n+nl}{\texttt{untrusted}}\PY{o}{\texttt{]}} and seal the
result, producing an error.

\begin{figure}[h]
\begin{BVerbatim}[commandchars=\\\{\}]
\textcolor{lightgray}{4} \PY{k}{let} \PY{o}{_} \PY{o}{:} \PY{o}{![}\PY{n+nl}{untrusted}\PY{o}{]} \PY{o}{[}\PY{n+nl}{untrusted}\PY{o}{]} \PY{k+kt}{bool} \PY{o}{=}
\textcolor{lightgray}{5}     \PY{k}{unseal} neural_net \PY{k}{as} f \PY{k}{in} \PY{k}{seal} \PY{o}{\{} f user_input \PY{o}{\}}
\end{BVerbatim}
\end{figure}

Note that if we had written the brackets inside the \PY{k}{\texttt{seal}} as
above, this would no longer error. As mentioned in \autoref{sec:order}, the
suspension associated to
\PY{o}{\texttt{[}}\PY{n+nl}{\texttt{untrusted}}\PY{o}{\texttt{]}} renders this
inert. The sealed application expression for \texttt{f} has not necessarily
been thrust into the presence of data violating its robustness
requirements---\texttt{user\_input}---because the brackets suspend the inner
computation and capture its dependencies before reaching \PY{k}{\texttt{seal}}.
Observe that one may still interact with this expression without introducing
its invalid final result by, say, \PY{k}{\texttt{unseal}}ing it without using
$\produceEx{\forceEx{-}}$ on the contents. This is not the case for the
original variant of lines 4 and 5.

Surprisingly, raising an error here is a choice left to the implementation, not
a matter of soundness. The examples in this section will return errors upon a
conflict being detected in the type of a top-level definition, which is a
matter of checking the indistinguishability condition from
\autoref{sec:simultaneous}. A practical implementation may choose to return
errors more proactively in order to place them closer to the offending
machinery. Whatever is done, the conflict is guaranteed to be statically
tracked---and therefore the user to be informed ahead-of-time---due to
soundness. On the current example soundness says the following, which we will
substantiate fully in the next section.

\begin{property}[Robustness, informally]
    \label{cor:untrusted}
    If we have expressions $\typing{\PY{n+nl}{\alpha}}{x : \PY{o}{\texttt{![}}\PY{n+nl}{\alpha}\PY{o}{\texttt{]}}~\PY{k+kt}{\texttt{t}}}{\varnothing}{v}{\PY{o}{\texttt{[}}\PY{n+nl}{\alpha}\PY{o}{\texttt{]}}~\PY{k+kt}{\texttt{t\textquotesingle}}}$ and $\typing{\PY{n+nl}{\alpha}}{\varnothing}{\PY{n+nl}{\alpha}}{v_1, v_2}{\PY{o}{\texttt{![}}\PY{n+nl}{\alpha}\PY{o}{\texttt{]}}~\PY{k+kt}{\texttt{t}}}$ then $\Sub{v_1}{x}{v}$ behaves equivalently to $\Sub{v_2}{x}{v}$.
\end{property}

Set $\PY{n+nl}{\alpha} = \PY{n+nl}{\texttt{untrusted}}$ and $x =
\texttt{neural\_net}$, and imagine $v$ as the term on lines 4-5. This corollary
states that the semantics of $v$ is independent of \texttt{neural\_net} by way
of the behavior of the former being invariant under differing implementations
of the latter. The intuition is that \texttt{neural\_net} is not robust against
\PY{o}{\texttt{[}}\PY{n+nl}{\texttt{untrusted}}\PY{o}{\texttt{]}} computations,
so for the program to be correct, the latter must remain ignorant of the
former. That is, $v_1$ could be the implementation from line 3 and $v_2$ the
constant function returning \PY{n+nd}{\texttt{true}}, and both programs would
be considered to behave equivalently. This is done by having program equality,
to be made rigorous in \autoref{sec:metatheory}, be \textit{up to
indistinguishability.} Indistinguishability trivializes equality upon conflict,
which has the effect of isolating the offending usage from the rest of the
program which does not misbehave. This means that it is not possible to use
data from lines 4-5 in a program at the type in line 6.
\textbf{\autoref{cor:untrusted} justifies that \textit{tracking} suffices over
\textit{erroring} to ensure that data cannot be misused by conflicting
computations.}

The previous point is worth reiterating: from a \textit{formal} point of view, it is of
no matter whether we return errors. Programs in conflict have derivations under
the core type system from \autoref{sec:syntax-typing}. What we focus on instead
is separating misbehaving parts of the program from those well-behaved---so a
program whose top-level type indicates no conflict cannot misbehave. We do this
because non-conflict is the focal point of our approach to information flow,
\textit{defining its conception of dependency tracking} rather than simply
being an additional reasoning facility of it. As such, its most natural
phrasing is by far as a matter of separation of conflicted and non-conflicted
parts, not the allowability of conflict itself. Significant theoretical
advantages will be gained from this perspective in \autoref{sec:metatheory}.

\subsubsection{Again, for Confidentiality}

\begin{figure}
\begin{BVerbatim}[commandchars=\\\{\}]
\textcolor{lightgray}{1} \PY{k}{let} make_user_profile \PY{o}{:} \PY{o}{![}\PY{n+nl}{secret}\PY{o}{] (}\PY{k+kt}{int} \PY{o}{*} \PY{k+kt}{string} \PY{o}{->} \PY{k+kt}{data}\PY{o}{)} \PY{o}{=} \PY{k}{seal} \PY{o}{...}
\end{BVerbatim}
\vspace{1em}
\textcolor{lightgray}{\hrule}
\vspace{1em}
\begin{BVerbatim}[commandchars=\\\{\}]
\textcolor{lightgray}{2} \PY{k}{let} \PY{o}{_} \PY{o}{:} \PY{o}{[}\PY{n+nl}{untrusted}\PY{o}{]} \PY{o}{![}\PY{n+nl}{secret}\PY{o}{]} \PY{k+kt}{\texttt{data}} \PY{o}{=}
\textcolor{lightgray}{3}     \PY{o}{\{} \PY{k}{unseal} make_user_profile \PY{k}{as} f \PY{k}{in} \PY{k}{seal} \PY{o}{(}f \PY{n+nd}{0} user_input\PY{o}{)} \PY{o}{\}}
\textcolor{lightgray}{4} \error{let _ : [secret] ![secret] data =}
\textcolor{lightgray}{5}     \error{\{ unseal make_user_profile as f in seal (f 0 password) \}} \danger\phantom{\hspace{2.25em}}
\end{BVerbatim}
\caption{Excluding Secrets from Public Data}
\label{fig:revisiting-confid}
\end{figure}

We have just reviewed an example of integrity reasoning, acting to prevent an
untrusted portion of the program from compromising a sensitive one. Let us
exercise the same pattern towards confidentiality.
\autoref{sec:isolating-info-intro} gestured at a scenario where secrets are
prevented from entering public settings. \autoref{fig:revisiting-confid} works
a full program in this spirit.

Specifically, we consider a fragment of the machinery for a program
implementing a social media program. We start on line 1 with
\texttt{make\_user\_profile}, which models functionality to create a
public-facing user profile. It takes an \PY{k+kt}{\texttt{int}} for the
creation timestamp and a \PY{k+kt}{\texttt{string}} for the username and
returns some \PY{k+kt}{\texttt{data}} representing this information. We use it
much as we did \texttt{neural\_net}, with that on line 3 passing it
\texttt{user\_input} from the prior example. This poses no issue here, as one
might expect. The usage on line 5 runs into an error, risking exposing password
data in a public setting. \autoref{cor:untrusted} states that its semantics is
independent of \texttt{make\_user\_profile}. This is due to the conflict in its
type giving way to indistinguishability, meaning we can freely swap out
\texttt{make\_user\_profile} for the constant function ignoring its inputs.
This again has the practical effect of isolating misuses of
\texttt{make\_user\_profile} from the rest of the program.

\subsubsection{Confidentiality and Integrity}

Note that this example of confidentiality reasoning is structurally identical
to the previous covering integrity, with each splitting specifications across
the two modalities in analogous ways. \textbf{Despite accounting for both, the
division of reasoning within our system is not the one between confidentiality
and integrity,} but between dependencies and anti-dependencies. Surprisingly,
based on the examples shown, it seems that \textit{half} of each of
confidentiality and integrity is a matter of dependencies such as
\PY{o}{\texttt{[}}\PY{n+nl}{\texttt{untrusted}}\PY{o}{\texttt{]}} or
\PY{o}{\texttt{[}}\PY{n+nl}{\texttt{secret}}\PY{o}{\texttt{]}}, and the other
half is a matter of \textit{robustness against} those dependencies, as with
\PY{o}{\texttt{![}}\PY{n+nl}{\texttt{untrusted}}\PY{o}{\texttt{]}} and
\PY{o}{\texttt{![}}\PY{n+nl}{\texttt{secret}}\PY{o}{\texttt{]}}. The latter two
might be written under a conventional theory of information flow as
\PY{o}{\texttt{[}}\PY{n+nl}{\texttt{trusted}}\PY{o}{\texttt{]}} and
\PY{o}{\texttt{[}}\PY{n+nl}{\texttt{public}}\PY{o}{\texttt{]}}, using our
syntax for dependencies, with each respectively being allowed to combine with
\PY{o}{\texttt{[}}\PY{n+nl}{\texttt{untrusted}}\PY{o}{\texttt{]}} and
\PY{o}{\texttt{[}}\PY{n+nl}{\texttt{secret}}\PY{o}{\texttt{]}} to yield
\PY{o}{\texttt{[}}\PY{n+nl}{\texttt{untrusted}}\PY{o}{\texttt{]}} and
\PY{o}{\texttt{[}}\PY{n+nl}{\texttt{secret}}\PY{o}{\texttt{]}}. That is, the
latter \textit{infect} the former. However, this is insufficient for our usage
of \PY{o}{\texttt{![}}\PY{n+nl}{\texttt{untrusted}}\PY{o}{\texttt{]}} and
\PY{o}{\texttt{![}}\PY{n+nl}{\texttt{secret}}\PY{o}{\texttt{]}}.

The outlined scheme works for reasoning about which dependencies we
\textit{have}, but not for which dependencies we are \textit{not allowed to
have}---perhaps unsurprisingly, given the implied absence of anti-dependencies.
Consider a function at type
\mbox{\PY{o}{\texttt{[}}\PY{n+nl}{\texttt{trusted}}\PY{o}{\texttt{]}}
\PY{k+kt}{\texttt{bool}} \PY{o}{\texttt{->}} \PY{o}{\texttt{...}}}, which
seemingly wishes to avoid taking untrusted input. Despite its argument type, it
can still be passed untrusted information via indirect flows, say, conditional
branching surrounding its call site. Merely specifying its argument to be
\PY{o}{\texttt{[}}\PY{n+nl}{\texttt{trusted}}\PY{o}{\texttt{]}} does not
suffice. This issue cannot easily be solved in the domain of dependencies
alone. We gain the ability to specify against untrusted inputs in reasoning
about \textit{robustness against} dependencies, which by definition accounts
for the surrounding context and its indirect flows. Prior work tends to either
force users to manually inspect portions of the program to be free of undesired
labels \cite[Fig. 4]{orbaek1997trust}, which scales poorly in the presence of
polymorphism as noted in \autoref{sec:isolating-flow-intro}, or introduces
ad-hoc specification features---like a function-level robustness label checked
against a special \textit{program counter} label denoting the current control
flow dependencies [\citealp[\S 2.8]{myers1999mostly}; \citealp[Fig.
8]{cecchetti2017nonmalleable}]. \textbf{We reify robustness restrictions as
their \textit{own} first-class dependency tracking connective, recognizing its
equal stature to the usual one towards obtaining full-spectrum information flow
reasoning.} This ecumenical view justifies the preceding deconstruction of
confidentiality and integrity.

\subsection{Robust Downgrading}
\label{sec:robust-downgrading}

Downgrading \textit{removes} dependencies from the quasi-open and quasi-closed
modalities, for instance, eliding the \PY{n+nl}{\valpha} from an
\PY{o}{\texttt{![}}\PY{n+nl}{\valpha} \PY{n+nl}{\vbeta}\PY{o}{\texttt{]}}
\PY{k+kt}{\texttt{int}} to get
\PY{o}{\texttt{![}}\PY{n+nl}{\vbeta}\PY{o}{\texttt{]}} \PY{k+kt}{\texttt{int}}.
Robustness and information flow typing have previously made contact in service
of downgrading \cite{cecchetti2017nonmalleable, zdancewic2001robust},
specifically \textit{robust declassification}. \textit{Declassification} is
downgrading applied toward confidentiality, or dependencies regarding secrets.
\textit{Robust} declassification recognizes the sensitivity of declassification
to manipulation by its environment, protecting against untrusted input. A
stronger variant of this feature can be accessed in our system as a special
case of the reasoning tools already reviewed, owing to their first-class
treatment of robustness. As a motivating example, let us write a variant of the
password checker in \autoref{fig:password-check} which uses (1) downgrading and
(2) robust declassification to resist brute force guess attempts. We step
through each point in turn via the program in
\autoref{fig:robust-declassification}.

\begin{figure}
\begin{minipage}{0.445\textwidth}
\centering
    \begin{BVerbatim}[commandchars=\\\{\}]
\textcolor{lightgray}{1}  \PY{k}{module} \PY{k+kt}{RobustDeclassify} \PY{o}{:} \PY{k}{sig}
\textcolor{lightgray}{2}    \PY{k}{dependency} \PY{n+nl}{rd}
\textcolor{lightgray}{3} 
\textcolor{lightgray}{4}    \PY{k}{val} declassify \PY{o}{:} \PY{o}{\texttt{[}}\PY{n+nl}{rd} \PY{n+nl}{\valpha}\PY{o}{\texttt{]}} \PY{k+kt}{t} \PY{o}{->}
\textcolor{lightgray}{5}        \PY{o}{\texttt{[}}\PY{n+nl}{\valpha}\PY{o}{\texttt{]}} \PY{o}{![}\PY{n+nl}{untrusted}\PY{o}{]} \PY{k+kt}{t}
\textcolor{lightgray}{6}  \PY{k}{end} \PY{o}{=} \PY{k}{struct}
\textcolor{lightgray}{7}    \PY{k}{dependency} \PY{n+nl}{rd} \PY{o}{=} \PY{o}{[]}
\textcolor{lightgray}{8}
\textcolor{lightgray}{9}    \PY{k}{let} declassify \PY{o}{=} fun arg \PY{o}{->}
\textcolor{lightgray}{10}       \PY{o}{\{} \PY{k}{seal} arg \PY{o}{\}}
\textcolor{lightgray}{11} \PY{k}{end}
\textcolor{lightgray}{12}
\textcolor{lightgray}{13} \PY{k}{open} \PY{k+kt}{RobustDeclassify}
    \end{BVerbatim}
\end{minipage}
\begin{minipage}{0.545\textwidth}
\centering
    \begin{BVerbatim}[commandchars=\\\{\}]
\textcolor{lightgray}{14} \PY{k}{module} \PY{k+kt}{PasswordChecker} \PY{o}{:} \PY{k}{sig}
\textcolor{lightgray}{15}   \PY{k}{dependency} \PY{n+nl}{secret}
\textcolor{lightgray}{16}                                                                                                                                                      
\textcolor{lightgray}{17}   \PY{k}{val} password \PY{o}{:} \PY{o}{[}\PY{n+nl}{secret}\PY{o}{]} \PY{k+kt}{string}
\textcolor{lightgray}{18}   \PY{k}{val} check \PY{o}{:} \PY{k+kt}{string} \PY{o}{->}
\textcolor{lightgray}{19}        \PY{o}{![}\PY{n+nl}{untrusted}\PY{o}{]} \PY{k+kt}{bool}
\textcolor{lightgray}{20} \PY{k}{end} \PY{o}{=} \PY{k}{struct}
\textcolor{lightgray}{21}   \PY{k}{dependency} \PY{n+nl}{secret} \PY{o}{=} \PY{o}{[}\PY{n+nl}{rd}\PY{o}{]}
\textcolor{lightgray}{22}
\textcolor{lightgray}{23}   \PY{k}{let} password \PY{o}{=} \PY{o}{\{} \PY{l+s}{"takver"} \PY{o}{\}}
\textcolor{lightgray}{24}   \PY{k}{let} check \PY{o}{=} \PY{k}{fun} attempt \PY{o}{->}
\textcolor{lightgray}{25}       declassify \PY{o}{(}attempt \PY{o}{==} password\PY{o}{)}
\textcolor{lightgray}{26} \PY{k}{end}
    \end{BVerbatim}
\end{minipage}
\caption{Robust Password Checking}
\label{fig:robust-declassification}
\end{figure}

We begin on the left with a \PY{k}{\texttt{module}}
\PY{k+kt}{\texttt{RobustDeclassify}}. It introduces a \PY{k}{\texttt{dependency}}
\PY{n+nl}{\texttt{rd}} and a single method \texttt{declassify} which takes an
argument with dependencies \mbox{\PY{o}{\texttt{[}}\PY{n+nl}{\valpha}
\PY{n+nl}{\texttt{rd}}\PY{o}{\texttt{]}}}, peeling off the
\PY{n+nl}{\texttt{rd}} dependency while introducing an anti-dependency against
\PY{n+nl}{\texttt{untrusted}} information. This can be understood as an
existential
$\existsTy{\PY{n+nl}{\texttt{rd}}}{\PY{o}{\texttt{[}}\PY{n+nl}{\texttt{rd}}~\PY{n+nl}{\alpha}\PY{o}{\texttt{]}}~\PY{k+kt}{\texttt{t}}~\PY{o}{\texttt{->}}~\PY{o}{\texttt{[}}\PY{n+nl}{\alpha}\PY{o}{\texttt{]}}~\PY{o}{\texttt{![}}\PY{n+nl}{\texttt{untrusted}}\PY{o}{\texttt{]}}~\PY{k+kt}{\texttt{t}}}$.
Its implementation proceeds, recalling \textsc{T-Pack}, by determining an
implementation of \PY{n+nl}{\texttt{rd}}---here the empty set of dependencies
\PY{o}{\texttt{[}}\PY{o}{\texttt{]}}. We then write the body of
\texttt{declassify} against \PY{o}{\texttt{[}}\PY{o}{\texttt{]}}, taking in an
\texttt{arg} at type \PY{o}{\texttt{[}}\PY{n+nl}{\valpha}\PY{o}{\texttt{]}}
\PY{k+kt}{\texttt{t}} due to the chosen implementation of
\PY{n+nl}{\texttt{rd}}. Notably, \PY{n+nl}{\texttt{rd}} is elided. The rest
comes easily, simply enclosing \texttt{arg} inside both modalities and
returning it. Finally, we \PY{k}{\texttt{open}}
\PY{k+kt}{\texttt{RobustDeclassify}} into the current scope.

Continuing to the right, we have the \PY{k+kt}{\texttt{PasswordChecker}}
interface. It introduces a \PY{k}{\texttt{dependency}} \PY{n+nl}{\texttt{secret}} used to
represent \texttt{password} data. The type of \texttt{password} is unchanged
from \autoref{fig:password-check}. However, the type of \texttt{check} differs,
first in that its \PY{k+kt}{\texttt{bool}} result no longer depends on
\PY{n+nl}{\texttt{secret}}, but also in that it is now anti-dependent on
\PY{n+nl}{\texttt{untrusted}}. The reason for this becomes clear as we look to
the implementation for \PY{k+kt}{\texttt{PasswordChecker}}. We implement
\PY{n+nl}{\texttt{secret}} as
\PY{o}{\texttt{[}}\PY{n+nl}{\texttt{rd}}\PY{o}{\texttt{]}}. We are then forced
to invoke \texttt{declassify} in the implementation of \texttt{check} in order
to elide the dependency of its body on \PY{n+nl}{\texttt{rd}}---from
\texttt{password}---rather than having it implicitly removed by virtue of being
empty as in \PY{k+kt}{\texttt{RobustDeclassify}}. So we get
\PY{o}{\texttt{![}}\PY{n+nl}{\texttt{untrusted}}\PY{o}{\texttt{]}} in the
return type. This restriction prevents \texttt{check} from being affected by
untrusted inputs. Summarizing, \PY{k+kt}{\texttt{RobustDeclassify}} downgrades
internally in order to export the \texttt{declassify} operation, while
\PY{k+kt}{\texttt{PasswordChecker}} uses it to robustly declassify.

Consider what \autoref{cor:untrusted} means for this example. When we use
\texttt{check} in an
\PY{o}{\texttt{[}}\PY{n+nl}{\texttt{untrusted}}\PY{o}{\texttt{]}} setting we
find ourselves in a situation analogous to that on lines 4-5 of
\autoref{fig:revisiting}. As there, indistinguishability provides that any uses
of the \PY{o}{\texttt{![}}\PY{n+nl}{\texttt{untrusted}}\PY{o}{\texttt{]}}
\PY{k+kt}{\texttt{bool}} in the return value of \texttt{check} can be freely
replaced with any other value of the same type. That is, under an
\PY{o}{\texttt{[}}\PY{n+nl}{\texttt{untrusted}}\PY{o}{\texttt{]}} environment,
it is as if the declassification never occurred! The declassified data
\textit{itself} is made indistinguishable, or freely replaceable with any
other; it is effectively un-leaked. Practically, robust declassifiers used with
untrusted data cannot affect the rest of the program; they are isolated, as
before.

\begin{figure}[h]
\begin{BVerbatim}[commandchars=\\\{\}]
\textcolor{lightgray}{27} \PY{k}{let} \PY{o}{_} \PY{o}{:} \PY{o}{[}\PY{n+nl}{untrusted}\PY{o}{]} \PY{k+kt}{string} \PY{o}{\texttt{->}} \PY{o}{[}\PY{n+nl}{untrusted}\PY{o}{]} \PY{k+kt}{bool} \PY{o}{=}
\textcolor{lightgray}{28}     \PY{k}{fun} user_input \PY{o}{->} rate_limit \PY{o}{\{} check user_input \PY{o}{\}}
\end{BVerbatim}
\end{figure}

We realistically \textit{do} want to be able to pass
\PY{o}{\texttt{[}}\PY{n+nl}{\texttt{untrusted}}\PY{o}{\texttt{]}} user input to
this password checker, so we are not through with the example. It will not do
to permanently carry around the
\PY{o}{\texttt{![}}\PY{n+nl}{\texttt{untrusted}}\PY{o}{\texttt{]}} restriction.
Recalling that we specifically wanted to avoid \textit{brute force guessing}
resulting from untrusted input, the right move might be to connect
downgrading---or removing---the
\PY{o}{\texttt{![}}\PY{n+nl}{\texttt{untrusted}}\PY{o}{\texttt{]}} restriction
to a rate limiting operation. Downgrading anti-dependencies works analogously
to dependencies, using a \PY{o}{\texttt{[]}} implementation of
\PY{n+nl}{\texttt{untrusted}}. Assume an \PY{k+kt}{\texttt{Untrusted}} module
exporting \PY{n+nl}{\texttt{untrusted}} and a function
\mbox{\texttt{rate\_limit} \PY{o}{\texttt{:}}
\PY{o}{\texttt{[}}\PY{n+nl}{\valpha}\PY{o}{\texttt{]}}
\PY{o}{\texttt{![}}\PY{n+nl}{\texttt{untrusted}}\PY{o}{\texttt{]}}
\PY{k+kt}{\texttt{t}} \PY{o}{\texttt{->}}
\PY{o}{\texttt{[}}\PY{n+nl}{\valpha}\PY{o}{\texttt{]}} \PY{k+kt}{\texttt{t}}}.
The idea with \texttt{rate\_limit} is that it runs the computation passed to it
and does not allow itself to be invoked above some frequency.
It is not dependency elided because it expects to take a suspended computation
as argument, so taking the argument at the open modality is convenient.
Assuming the \PY{k+kt}{\texttt{PasswordChecker}} module is
\PY{k}{\texttt{open}}ed, we have the above. Only the
\PY{o}{\texttt{[}}\PY{n+nl}{\texttt{untrusted}}\PY{o}{\texttt{]}} from
\texttt{user\_input} remains. Note with $\PY{n+nl}{\alpha} =
\PY{n+nl}{\texttt{untrusted}}$, \texttt{rate\_limit} can downgrade data
\textit{in conflict} to data \textit{out of conflict}.
\textbf{Downgrading precludes us from always rejecting programs with
conflicts.} \autoref{sec:metatheory} will say more.

We impose here the requirement that rate limiting take place while granting
clients control over its granularity. This might be used to support some number
of rapid attempts before a timed lock-out, or perhaps logging into multiple
devices simultaneously. One must only ensure the argument to
\texttt{rate\_limit} does not already use brute force. A final modularity issue
remains.

\subsubsection{Delegating Conflicts}

It is slightly odd for \PY{k+kt}{\texttt{Untrusted}} to offer an operation
relevant largely to the \PY{k+kt}{\texttt{PasswordChecker}}. Avoiding brute
force may not be the definition of robustness against untrusted input
appropriate for other modules. To fix this, we introduce a new \PY{k}{\texttt{dependency}}
\PY{k+kt}{\texttt{PasswordChecker}}\PY{o}{\texttt{.}}\PY{n+nl}{\texttt{untrusted}}
implemented as \PY{o}{\texttt{[]}}, changing \texttt{check} to use it.
\texttt{rate\_limit} can then be made part of the interface of
\PY{k+kt}{\texttt{PasswordChecker}}, downgrading
\PY{k+kt}{\texttt{PasswordChecker}}\PY{o}{\texttt{.}}\PY{n+nl}{\texttt{untrusted}}.
So \texttt{rate\_limit} can \textit{only} be used to remove the restriction
against \PY{n+nl}{\texttt{untrusted}} from \texttt{check}.

But how is
\PY{k+kt}{\texttt{PasswordChecker}}\PY{o}{\texttt{.}}\PY{n+nl}{\texttt{untrusted}}
related to
\PY{k+kt}{\texttt{Untrusted}}\PY{o}{\texttt{.}}\PY{n+nl}{\texttt{untrusted}}?
We introduce a special operation \texttt{avoid} \PY{o}{\texttt{:}}
\PY{o}{\texttt{![}}\PY{k+kt}{\texttt{PasswordChecker}}\PY{o}{\texttt{.}}\PY{n+nl}{\texttt{untrusted}}\PY{o}{\texttt{]}}
\PY{k+kt}{\texttt{t}} \PY{o}{\texttt{->}}
\PY{o}{\texttt{![}}\PY{k+kt}{\texttt{Untrusted}}\PY{o}{\texttt{.}}\PY{n+nl}{\texttt{untrusted}}\PY{o}{\texttt{]}}
\PY{k+kt}{\texttt{t}} to the former module. Recall that when leaving the scope
of an existential, the third and fourth premises of \textsc{T-Open} from
\autoref{fig:computations} preclude mentioning the to-be-unbound dependency
variable. This means that a program using a dependency from some module leaving
scope must find a way to be rid of that dependency. Downgrading provides one
path, and \texttt{avoid} another. The intuition behind the latter is to reify
the meaning of a dependency going out of scope in terms of one still in scope.
A practical implementation would introduce syntax for generating these
\texttt{avoid} declarations, such as \PY{k}{\texttt{dependency}}
\PY{n+nl}{\texttt{untrusted}} \PY{k}{\texttt{conflicts}}
\PY{k+kt}{\texttt{Untrusted}}\PY{o}{\texttt{.}}\PY{n+nl}{\texttt{untrusted}}.
The generated declarations would be automatically called when leaving scope,
and the type checker could further use them to proactively raise errors between
conflicted dependencies. Alternate solutions to this issue, left as future
work, fall under the purview of the \textit{avoidance problem}; \citet[\S
1]{crary2020focused} gives an overview.

Under this setup, each module controls the downgrading of its own restrictions
while still referencing the global notion of that restriction. Existential
scoping as highlighted here also prevents existentials from perpetually hiding
conflicts, say, by abstracting the same dependency behind different variables.
It will force either their declassification or their \texttt{avoid}ance.

\subsubsection{Comparing to Standard Robust Downgrading, pt. 1}

We mentioned our approach exhibits a stronger variant of the robustness
reasoning explored previously by \citet{zdancewic2001robust}. Their treatment
of the preceding password checking example reveals why. Under standard robust
declassification, \texttt{declassify} would be a built-in operation of the
language imposing its own ad-hoc robustness restrictions against
\PY{o}{\texttt{[}}\PY{n+nl}{\texttt{untrusted}}\PY{o}{\texttt{]}}. These
restrictions would be satisfied by simply using \texttt{endorse} on
\texttt{attempt} before passing \mbox{\texttt{attempt} \PY{o}{\texttt{==}}
\texttt{password}} to \texttt{declassify} \cite[Fig.
1]{cecchetti2017nonmalleable}, where \textit{endorsement} acts to downgrade
\PY{o}{\texttt{[}}\PY{n+nl}{\texttt{untrusted}}\PY{o}{\texttt{]}} dependencies.
One might consider this implicit discharging of robustness restrictions for
trusted input a convenience, but given the preceding example, it is in reality
quite dangerous! Merely endorsing the attempt \PY{k+kt}{\texttt{string}} does
little to increase the safety of the ensuing declassification. That
declassification is safe from manipulation by untrusted input must be
\textit{independently} and \textit{explicitly} validated, as with
\texttt{rate\_limit}. Observe, however, that \texttt{rate\_limit} does
something quite distinct from endorsement: it downgrades the \textit{robustness
requirement itself}. This is possible because \textbf{our system possesses
\textit{internal} and therefore \textit{programmable} robustness reasoning
through the closed modality,} not specially attached to declassification
behavior or the like. In turn, we need not connect its relaxation to inbuilt
operations like \texttt{endorse}, instead deferring control to the interface-
or specification-writer. 

Beyond robust declassification, we can express robust \textit{endorsement} in
much the same way, producing restrictions on its output. One form of
this is \textit{transparent endorsement}, inducing
\PY{o}{\texttt{![}}\PY{n+nl}{\texttt{secret}}\PY{o}{\texttt{]}} restrictions.
Another variant inducing
\PY{o}{\texttt{![}}\PY{n+nl}{\texttt{untrusted}}\PY{o}{\texttt{]}} has not been
realized by prior work, despite its usefulness, because it risks breaking the
confidentiality-integrity duality. We do not mind doing so, because \textbf{our
duality is between dependencies and robustness against them.}

A few remarks follow. First, our specifications tend to simplify those of prior
work, where each annotation must carry separate confidentiality and integrity
components governed by opposing lattice operations \cite{arden2016calculus}.
Second, we have used verbose syntax throughout these examples to relate them
more closely to \autoref{sec:syntax-typing}. A practical implementation could
extend the \mbox{\PY{o}{\texttt{\{}\texttt{...}\texttt{\}}}} notation to the
closed modality, automatically unsealing and resealing data under it at the
bracket boundary. Bidirectional synthesis of the brackets would account for the
syntax of the introduction. We leave this to future work. Third, we speak here
in terms of dependencies and anti-dependencies since the call-by-value
fragment, by blurring polar behavior, resists discussion in terms of
information and flows. We now look to the metatheory of our duality which,
among other developments, will allow us to substantiate
\autoref{cor:untrusted}. This will be significantly more straightforward than
prior accounts of robustness.

\section{Metatheory}
\label{sec:metatheory}

\begin{figure}
\begin{align*}
    \exact{e}{e'}{\fTy{A}}{\Delta}{\phi} \triangleq&~ e \mapsto^* \retEx{v},\, e' \mapsto^* \retEx{v'},\, \exact{v}{v'}{A}{\Delta}{\phi} \\
    \exact{e}{e'}{\arrTy{A}{\ul{B}}}{\Delta}{\phi} \triangleq&~ \Delta \subseteq \Delta',\, \phi'_\open \vdash_{\textcolor{gray}{(\supseteq)}} \phi_\open,\, \exact{v}{v'}{A}{\Delta'}{\phi'} \implies \\
                                                             &~ \exact{\apEx{e}{v}}{\apEx{e'}{v'}}{\ul{B}}{\Delta'}{\phi'} \\
    \exact{e}{e'}{\forallTy{\PY{n+nl}{\alpha}}{\ul{B}}}{\Delta}{\phi} \triangleq&~ \Delta \subseteq \Delta',\, \Delta' \vdash \phi' \implies \\
        &~ \exact{\instEx{e}{\phi'}}{\instEx{e'}{\phi'}}{\Sub{\phi'}{\PY{n+nl}{\alpha}}{\ul{B}}}{\Delta'}{\phi} \\
    \exact{e}{e'}{\opTy{\phi'}{\ul{B}}}{\Delta}{\phi} \triangleq&~ \exact{\produceEx{e}}{\produceEx{e'}}{\ul{B}}{\Delta}{\phi\, \phi'} \\
    \exact{v}{v'}{\uTy{\ul{B}}}{\Delta}{\phi} \triangleq&~ \exact{\forceEx{v}}{\forceEx{v'}}{\ul{B}}{\Delta}{\phi} \\
    \exact{\unitEx}{\unitEx}{\unitTy}{\Delta}{\phi} \triangleq&~ \mathsf{true} \\
    \exact{\injlEx{v}}{\injlEx{v'}}{\sumTy{A_1}{A_2}}{\Delta}{\phi} \triangleq&~ \exact{v}{v'}{A_1}{\Delta}{\phi} \\
    \exact{\injrEx{v}}{\injrEx{v'}}{\sumTy{A_1}{A_2}}{\Delta}{\phi} \triangleq&~ \exact{v}{v'}{A_2}{\Delta}{\phi} \\
    \exact{\pairEx{v_1}{v_2}}{\pairEx{v_1'}{v_2'}}{\pairTy{A_1}{A_2}}{\Delta}{\phi} \triangleq&~ \exact{v_1}{v_1'}{A_1}{\Delta}{\phi},\, \exact{v_2}{v_2'}{A_2}{\Delta}{\phi} \\
    \exact{\packEx{\phi'}{v}}{\packEx{\phi'}{v'}}{\existsTy{\PY{n+nl}{\alpha}}{A}}{\Delta}{\phi} \triangleq&~ \Delta \vdash \phi',\, \exact{v}{v'}{\Sub{\phi'}{\PY{n+nl}{\alpha}}{A}}{\Delta}{\phi} \\
    \exact{\sealEx{\phi'}{v}}{\sealEx{\phi'}{v'}}{\clTy{\phi'}{A}}{\Delta}{\phi} \triangleq&~ \phi_\open \vdash \phi'_\closed,\, \exact{v}{v}{A}{\Delta}{\phi},\, \exact{v'}{v'}{A}{\Delta}{\phi} \mathrel{\mathsf{or}} \\
        &~ \exact{v}{v'}{A}{\Delta}{\phi}
\end{align*}
\caption{Logical Relation for Robust Non-Interference}
\label{fig:non-interference}
\end{figure}

Soundness for our system will take the form of an indistinguishability or
\textit{non-interference} property. We capture it through a binary logical
relation, shown in \autoref{fig:non-interference}. Following the type system in
\autoref{sec:syntax-typing}, it is bifurcated into two mutually recursive
sub-relations for values and computations. This relation will define
program equivalence as used by \autoref{cor:untrusted}. All results in this
section have been mechanized in Lean; proof scripts are provided in the
supplemental materials.

Starting with that for computations, the form of the relation is
$\exact{e}{e'}{\ul{B}}{\Delta}{\phi}$, read as ``$e$ is semantically related to
$e'$ at computation type $\ul{B}$ and ambient dependencies $\phi$, all scoped
under $\Delta$.'' The first case of the relation, at $\fTy{A}$, evaluates both
sides to $\retEx{v}$ and checks their membership in the value relation. The
evaluation relation used is entirely standard, so we elide it. This case gives
rise to the following two properties stating closure under forward and reverse
evaluation.

\begin{lemma}[Head Reduction]
    If $\exact{e_1}{e_2}{\ul{B}}{\Delta}{\phi}$ and $e_1 \evals^* e_1'$ then $\exact{e_1'}{e_2}{\ul{B}}{\Delta}{\phi}$.
\end{lemma}

\begin{proof}
    By induction on $\ul{B}$ and use of evaluation in the case for $\fTy{A}$.
\end{proof}

\begin{lemma}[Head Expansion]
    If $\exact{e_1}{e_2}{\ul{B}}{\Delta}{\phi}$ and $e_1' \evals^* e_1$ then $\exact{e_1'}{e_2}{\ul{B}}{\Delta}{\phi}$.
\end{lemma}

\begin{proof}
    By induction on $\ul{B}$ and use of evaluation in the case for $\fTy{A}$.
\end{proof}

Beyond this case, negative types are defined by their elimination forms, so the
relation for computations is defined by eliminating the expressions being
related. For instance, the definition of relatedness at the open modality
invokes $\produceEx{-}$ on each side, raising the ambient security level, and
terms at function type are related extensionally. Note that when the latter
assumes related arguments at $A$, it is at a raised ambient security level
$\phi'$. This is required to satisfy the following properties. $t$ and
$\mathbb{T}$ respectively denote terms and types generic over values and
computations.

\begin{lemma}[Dependency Monotonicity]
    \label{lem:dependency-mono}
    If $\exact{t}{t'}{\mathbb{T}}{\Delta}{\phi}$ and $\phi'_\open
    \vdash_{\textcolor{gray}{(\supseteq)}} \phi_\open$ then
    $\exact{t}{t'}{\mathbb{T}}{\Delta}{\phi'}$.
\end{lemma}

\begin{proof}
    By induction on $\mathbb{T}$, using the further assumption $\phi'_\open
    \vdash \phi_\open$ in the case for $\arrTy{A}{\ul{B}}$.
\end{proof}

\begin{lemma}[Type Monotonicity]
    If $\exact{t}{t'}{\mathbb{T}}{\Delta}{\phi}$ and $\mathbb{T} \subseteq \mathbb{T'}$
    then $\exact{t}{t'}{\mathbb{T'}}{\Delta}{\phi}$.
\end{lemma}

\begin{proof}
    By induction on $\mathbb{T}$ and application of \autoref{lem:dependency-mono} in the case for $\opTy{\phi'}{\ul{B}}$.
\end{proof}

This means we can freely upgrade the relation along the subtyping relationship
from \autoref{sec:syntax-typing}. Inspecting the relation, the cases for the
computation types $\opTy{\phi'}{\ul{B}}$ and $\arrTy{A}{\ul{B}}$ are the only
ones modifying the ambient dependencies, as promised in
\autoref{sec:simultaneous}. A further monotonicity property arises from the
cases for functions and universal quantifiers extending the environment
$\Delta$.

\begin{lemma}[Kripke Monotonicity]
    If $\exact{t}{t'}{\mathbb{T}}{\Delta}{\phi}$ and $\Delta \subseteq \Delta'$
    then $\exact{t}{t'}{\mathbb{T}}{\Delta'}{\phi}$.
\end{lemma}

\begin{proof}
    By induction on $\mathbb{T}$, using the assumption $\Delta \subseteq \Delta'$ in the cases for $\arrTy{A}{\ul{B}}$ and $\forallTy{\PY{n+nl}{\alpha}}{\ul{B}}$.
\end{proof}

The above properties establish $\phi$ and $\Delta$ as \textit{Kripke worlds}
within the relation, with $\phi'_\open \vdash \phi_\open$ and $\Delta \subseteq
\Delta'$ known as \textit{Kripke quantification}. Importantly, the dependencies
\PY{n+nl}{\valpha} contained in $\Delta$ are exactly \textit{included middles}
from a semantic perspective, as noted in \autoref{sec:versus}. They will be
those relevant to indistinguishability as deployed by the relation.

Moving to the value relation $\exact{v}{v'}{A}{\Delta}{\phi}$, it is read as
``$v$ is semantically related to $v'$ at value type $A$ and ambient
dependencies $\phi$, all scoped under $\Delta$.'' Positive types are defined by
their introduction forms, so the value relation---but for $\uTy{\ul{B}}$---is
defined by pattern matching the values to be related with the introduction
forms expected of their types. Values not of the expected form are considered
unrelated. Any values inside matching introduction forms are further required
to be related. For instance, the case for pairs $A_1 \otimes A_2$ requires both
sides match the pattern $\pairEx{\_}{\_}$ and that their components be
pointwise related at $A_1$ and $A_2$.

The exception to the requirement that inner values be related occurs in the
all-important case for the quasi-closed modality $\clTy{\phi'}{A}$. There are
\textit{two} paths for values $v, v'$ to be related under it: (1) to satisfy
$\phi_\open \vdash \phi'_\closed$ with $v$ and $v'$ \textit{reflexively
related} at the inner type, or (2) for $v$ and $v'$ to simply be related at the
inner type as usual. The former uses the entailment $\phi_\open \vdash
\phi'_\closed$, mentioned in \autoref{sec:simultaneous} as the
indistinguishability condition, to relate values under the modality
nearly vacuously. Reflexive relatedness merely ensures that each value still
individually behaves according to its type. The property resulting from this
case relies on certain auxiliary lemmas.

\begin{lemma}[Symmetry]
    \label{thm:symm}
    If $\exact{t}{t'}{\mathbb{T}}{\Delta}{\phi}$ then $\exact{t'}{t}{\mathbb{T}}{\Delta}{\phi}$.
\end{lemma}

\begin{proof}
    By induction on $\mathbb{T}$.
\end{proof}

\begin{lemma}[Transitivity]
    \label{thm:trans}
    If $\exact{t}{t'}{\mathbb{T}}{\Delta}{\phi}$ and $\exact{t'}{t''}{\mathbb{T}}{\Delta}{\phi}$ then $\exact{t}{t''}{\mathbb{T}}{\Delta}{\phi}$.
\end{lemma}

\begin{proof}
    By induction on $\mathbb{T}$, using \autoref{thm:symm} for reflexivity in the $\arrTy{A}{\ul{B}}$ case.
\end{proof}

We use these to state the \textit{sealing lemma}, which captures the essence of
indistinguishability.

\begin{lemma}[Seal]
    \label{lem:seal}
    If \textcolor{Red}{$\exact{t}{t}{\mathbb{T}}{\Delta}{\phi}$, $\exact{t'}{t'}{\mathbb{T}}{\Delta}{\phi}$,} \textcolor{Purple}{$\mathbb{T} \geq \phi'$,} and \textcolor{Blue}{$\phi_\open \vdash \phi'_\closed$} then $\exact{t}{t'}{\mathbb{T}}{\Delta}{\phi}$.
\end{lemma}

\begin{proof}
    By induction on $\mathbb{T}$, using \autoref{thm:symm} and \autoref{thm:trans} for reflexivity in the $\arrTy{A}{\ul{B}}$ case.
\end{proof}

This lemma relates two arbitrary terms $t$ and $t'$ which
\textcolor{Red}{\textbf{(1)} satisfy the relation reflexively}
\textcolor{Purple}{\textbf{(2)} at a type $\mathbb{T}$ sealed at $\phi'$}
\textcolor{Blue}{\textbf{(3)} under ambient dependencies $\phi$ where
$\phi_\open \vdash \phi'_\closed$.} In effect, we can freely equate terms where
all the information in them is guaranteed to be sealed at some $\phi'$
conflicting with the ambient dependencies $\phi$. Relatedness being trivialized
for terms satisfying certain conditions prevents terms not satisfying those
conditions from depending on them while remaining in the relation. The method
by which the former are related is unavailable to the latter. In turn,
terms satisfying the antecedents of \autoref{lem:seal} are isolated from the
rest and cannot affect their behavior. This isolation property is known as
\textit{non-interference} \cite{goguen1982security}. That the isolated terms
are those violating their robustness requirements\footnote{Excepting
those vacuously satisfying $A \geq \phi'$ by containing no information, like
\texttt{unit}.} means that \textbf{non-interference is defined by robustness
within our system.} Misbehaving terms cannot interfere with those well-behaved.

Being derived from the semantics of the open and closed modalities,
\autoref{lem:seal} relies crucially on the case for closed being aware of the
ambient dependencies---specifically the value-level \textit{flows-to-be}---as
noted in \autoref{sec:syntax-typing-flows}. Importantly, this case is the
\textit{only} one making any concessions to indistinguishability; all other
types have a typical treatment of equality. \textbf{Our definition of program
equality does not hold indistinguishability in special regard. It derives from
the definition of equality at the closed modality under the ambient
dependencies of the open modality.} The isolation behavior we exploit in this
work arises inevitably from here.

Finally, observe that \autoref{lem:seal} forces the faithful tracking of both
$\opTy{\phi}{\ul{B}}$ and $\clTy{\phi'}{A}$ within the semantics. Failing to do
so would allow for terms vacuously related by it which do not actively satisfy
its antecedents. However, the relation would then require such terms to be truly
related according to their type, which is impossible in general. As an aside,
note that \autoref{lem:seal} embodies the equation from
\autoref{sec:simultaneous}. And through another lens, it states that if
$\mathbb{T}$ is $\closed_{\phi'}$-modal then it is $\open_{\phi}$-connected,
corresponding to the forward direction of \citet[Lem.
3.6]{sterling2021logical}.

A minor limitation is that the case for existentials requires both sides to
have the same implementing dependencies $\phi'$. We might relax this
restriction by interpreting dependency variables as relations to be substituted
into the type $A$, following the interpretation of type variables under
relational parametricity \cite{reynolds_1984}. Existentials could then be
included in the sealing judgement $A \geq \phi$, as forecast in
\autoref{sec:syntax-typing-info}. However, this induces a steep increase in
complexity and departs from all other accounts of non-interference. We leave
such a generalization to the future.

The relation so far has been defined on terms closed in term variables $x$. The
type system from \autoref{sec:syntax-typing} works on open terms. To bridge the
divide, we must generalize the former to open terms. Assume contexts $\Delta$,
$\Delta'$, and $\Gamma$, and assume $\phi$, $\phi'$ such that $\Delta \vdash
\phi$ and $\Delta' \vdash \phi'$. Define a map $\delta$ from dependency
variables $\PY{n+nl}{\alpha} \in \Delta$ to $\Delta' \vdash \phi$, and a map
$\gamma$ from term variables $x : A \in \Gamma$ to values $v$ closed in term
variables and scoped under $\Delta'$. Then define $\gamma \sim \gamma'$ such
that $\forall x : A \in \Gamma \have
\exact{\gamma(x)}{\gamma'(x)}{\delta(A)}{\Delta'}{\phi'}$, that is, mapping the
same variable to related values. When $\delta$ and $\gamma$ are passed types
and terms, assume they recurse through their structure, substituting for all
variables in their domain. The definition of the open relation is that under
the preceding assumptions, the following must hold:
$\exact{\delta(\gamma(t))}{\delta(\gamma'(t'))}{\delta(\mathbb{T})}{\Delta'}{\delta(\phi)
\cup \phi'}$. We write this concisely as
$\exactOpen{\Delta}{\Gamma}{\phi}{t}{t'}{\mathbb{T}}{\Delta'}{\phi'}$. We can
now state the fundamental theorem.

\begin{theorem}[Fundamental Theorem]
    \label{thm:ftlr}
    If $\typing{\Delta}{\Gamma}{\phi}{t}{\mathbb{T}}$ then
    $\exactOpen{\Delta}{\Gamma}{\phi}{t}{t}{\mathbb{T}}{\Delta'}{\phi'}$.
\end{theorem}

\begin{proof}
    By induction on a derivation of
    $\typing{\Delta}{\Gamma}{\phi}{t}{\mathbb{T}}$, using \autoref{lem:seal} in
    the \textsc{T-Unseal} case.
\end{proof}

This property connects static well-typedness to membership in the semantics. It
states that any well-typed program is reflexively related under the open
relation. This is nearly the same as being reflexively related under the closed
relation, except that each side receives \textit{related} rather than
\textit{equivalent} substitutions for its term variables. This is important for
the final corollary.

\begin{corollary}[Robust Non-Interference]
    \label{cor:rni}
    If $\typing{\Delta}{x : A}{\phi}{t}{\mathbb{T}}$ and $A \geq \phi$ and
    $\typing{\Delta}{\varnothing}{\phi}{v_1, v_2}{A}$ with $\phi$ nonempty then
    $\exact{\Sub{v_1}{x}{t}}{\Sub{v_2}{x}{t}}{\mathbb{T}}{\Delta}{\phi}$.
\end{corollary}

\begin{proof}
    Immediate from \autoref{thm:ftlr}, using \autoref{lem:seal} to obtain that $v_1$ and $v_2$ are related.
\end{proof}

Set $\Delta = \PY{n+nl}{\alpha}$, $\phi = \PY{n+nl}{\alpha}$, $A =
\clTy{\PY{n+nl}{\alpha}}{\ldots}$, and $\mathbb{T} =
\uTy{\opTy{\PY{n+nl}{\alpha}}{\fTy{\ldots}}}$ to obtain
\autoref{cor:untrusted}, using \autoref{fig:non-interference} as the definition
of equivalence. We elucidate this result in two parts. First, we briefly recap
the reconciliation of declassification and non-interference by
\citet{gouni2025structural} and review how this result extends it. We then
compare to the standard definition of robust downgrading.

\subsection{Downgrading for Cheap}

Where do we account for downgrading behavior? Notice in the definition of the
open relation for \autoref{thm:ftlr} that there are \textit{two} dependency
variable environments $\Delta$ and $\Delta'$ at hand. The first is the
environment coming from the static typing, and the second is that intended for
use in the closed relation. When we close the relation under $\Delta'$ as
there, we fix the dependencies relevant to the condition $\phi_\open \vdash
\phi'_\closed$, and therefore to non-interference. The definition of the open
relation then assumes related substitutions $\gamma \sim \gamma'$ for the term
variables in $\Gamma$, which has the effect of assuming
non-interference-preserving implementations of all in-scope downgraders.

For instance, we may pick \PY{n+nl}{\texttt{untrusted}} to be in $\Delta'$ with
$\delta$ mapping \PY{n+nl}{\texttt{untrusted}} to
$\PY{n+nl}{\texttt{untrusted}}; \epsilon$, making it able to be used to trigger
indistinguishability, and elide all other currently in-scope dependency
variables by mapping them to $\epsilon$. And with \mbox{\texttt{rate\_limit}
\PY{o}{\texttt{:}} \PY{o}{\texttt{[}}\PY{n+nl}{\valpha}\PY{o}{\texttt{]}}
\PY{o}{\texttt{![}}\PY{n+nl}{\texttt{untrusted}}\PY{o}{\texttt{]}}
\PY{k+kt}{\texttt{t}} \PY{o}{\texttt{->}}
\PY{o}{\texttt{[}}\PY{n+nl}{\valpha}\PY{o}{\texttt{]}} \PY{k+kt}{\texttt{t}}}
in scope we assume related implementations for it at its type. When downgrading
\PY{o}{\texttt{![}}\PY{n+nl}{\texttt{untrusted}}\PY{o}{\texttt{]}}, one might
think we risk bringing the input data out of conflict, which would be dangerous
for the reasons previously outlined. Since we have \textit{assumed} related
implementations of \texttt{rate\_limit}, however, we can be guaranteed related
outputs from related inputs. That is, \texttt{rate\_limit} and \textit{only}
\texttt{rate\_limit} is permitted to downgrade
\PY{o}{\texttt{![}}\PY{n+nl}{\texttt{untrusted}}\PY{o}{\texttt{]}} data while
preserving relatedness.

The idea of non-interference we have conveyed is thus \textit{already} stated
up to downgrading, without extra effort on our part. The shape of downgrading
is determined by which dependencies $\delta$ maps into $\Delta'$, and the
operations for which $\gamma \sim \gamma'$ provides related implementations.
Our work recognizes that this pattern continues to function when both the
quasi-open and quasi-closed modalities are in play, rather than just the former
as in \citet{gouni2025structural}. We defer to them for further exposition.
Note that when we transform a dependency to $\epsilon$, it is downgraded
because $\epsilon_\open \not\vdash \epsilon_\closed$ by \autoref{fig:pairing};
it can no longer be used to trigger indistinguishability, or as a justification
to trivialize equality. The non-emptiness assumption in \autoref{cor:rni} is
due to this. Further, when we fix the dependencies in $\Delta'$ these can be
thought of as the \textit{included middle} truth values interpreted in
\autoref{sec:versus}.

\subsection{Comparing to Standard Robust Downgrading, pt. 2}

Let us review what this all means for the conception of robustness from prior
work. Previous accounts of robustness have cast it as a
\textit{4-hyperproperty}, or a matter of 4 related traces of a program. Our
binary logical relation, however, only concerns 2 traces. The reason for the
simplification is due to our \textit{explicit} and \textit{typed} treatment of
robustness. We offer justification through the words of prior work in robust
declassification.

\begin{quote}
    If we could precisely identify the release events, this would allow us to
    specify robust declassification as a 2-safety property on those release
    events.
    \hfill \citet{cecchetti2017nonmalleable}
\end{quote}

Conventional work integrating robustness into information flow specifically
focuses on robust \textit{declassification} and its release of secrets.
Declassifiers are primitive to the language, and they are the focus of
robustness. As such, stating robust declassification requires first detecting
the use of declassification, and subsequently ensuring the robustness of its
uses. However, as shown in \autoref{sec:robust-downgrading}, we defer control
over robustness to authors of module interfaces. Specifically, we do this by
reifying robustness as its own connective, the closed modality. So the type of
the program directly states its robustness requirements, allowing our semantics
to be formulated around it.

Beyond robustness being made explicit in the type system, the \textit{typed}
nature of our semantics is the second key piece of the simplification. Prior
approaches to the semantics of robust declassification attempt to do so in an
\textit{untyped} way, purely reasoning over traces of programs. We take the
perspective that \textbf{the semantics of robustness is driven by types,}
giving our semantic definition access to the type of the computation over which
it reasons. Taking explicit specifications and typed semantics together,
robustness may be defined by interpreting the closed modality. This results in
a system which makes strides in expressiveness while presenting a significantly
simpler metatheory. \autoref{cor:rni} provides that data with robustness
restrictions against certain contexts is unavailable to computations in those
contexts, which is analogous to the separation properties pursued by the more
complex formulations. The above argument also applies to transparent
endorsement.

\subsection{Mechanization Notes}

Proofs of the theorems just stated can be found in
\texttt{Mechanize/Semantics.lean}, in the attached artifact. We occasionally
use set-theoretic subsetting and intersection operations in place of
entailments on $\phi_\open, \phi_\closed$ for ease of mechanization, where the
two are proven equivalent. We used free large language models to search the
Lean and Mathlib documentation. They were also used to complete routine cases
and fill holes in auxiliary lemmas while mechanizing the paper proofs of
\citet{gouni2025structural}, which ultimately provided helpful structure for
this development. We otherwise avoided the use of such tools in order to
extract the maximal metatheoretic insights.

\section{Related Work}
\label{sec:related}

Related work has structured our exposition throughout. We review here the few
remaining points.

\subsubsection*{Graded Effects}

Our framework may be understood as one of \textit{graded monads} or
\textit{effects}, grading $\opTy{\phi}{A}$ and $\clTy{\phi}{A}$ with the free
semilattice $\phi$. Of note, \citet{torczon2024cbpv} sets up graded effects in
a Call-By-Push-Value setting. There, the judgemental effect $\phi$ grades the
$\uTy{\ul{B}}$ connective and resides only within computations. This would be
the wrong move in our case because the open modality is \textit{idempotent},
meaning it is of no object how many time the open `effect' runs. So we need not
restrict its $\phi$ to the computation judgement; indeed, this would break the
closed modality.

\subsubsection*{Open and Closed}

We mentioned the forward direction of \citet[Lem. 3.6]{sterling2021logical}
corresponds to \autoref{lem:seal}. What about the reverse? This does not hold
in our case, due to the choices made in \autoref{sec:versus} to diverge from
the standard modalities. Fortunately, it is also not needed for our current
goals. A \textit{singleton} dependency context, in the style of the enriched
effect calculus \cite{egger2014enriched}, seems to be able to recover the full
theorem; we consider it future work.

Other work with the open and closed modalities has exploited the idea of
\textit{included middles} towards declassification-like behavior, if
unacknowledged. The \textit{abstraction function} of
\citet{grodin2026abstraction} appears to operate distantly similarly. More work
is needed to determine the relationship.

\subsubsection*{Confidentiality and Integrity}

The best-known duality within information flow is, of course, that between
confidentiality and integrity, introduced by \citet{biba1977integrity}.
\citet{cecchetti2017nonmalleable} maintains it in the presence of advanced
mechanisms like robust downgrading and transparent endorsement, though at the
cost of metatheoretic complexity. Future work may explore a translation of this
perspective into our system, solidifying the relationship between the two.

\subsubsection*{Substructurality}

\citet{gouni2026security} works in a system which extends
\citet{gouni2025structural} to be able to speak about $\phi$s which are not
sets. Specifically, they are able to speak about the order and quantity of
dependencies within the quasi-open modality. Whether the quasi-closed modality
can be made substructural in the same sense as done for the quasi-open modality
there is future work.

\section{Conclusion}
\label{sec:conclusion}

We have offered a number of different perspectives on the duality at the heart
of information flow, identifying with it several surprising and useful
phenomena spanning theory and practice. The type-theoretic consonance of our
analysis will likely continue to bear fruit beyond what has been discussed.
Shall secure information flow focus on tracking \textit{information} or
\textit{flows}? We say both.

\subsection*{Acknowledgements}

We thank Corinthia Aberl\'e, Harrison Grodin, and Jonathan Sterling for their
guidance in grasping the background material for this work. This research was
supported by the Department of Defense under Grant No. H98230-23-C-0275 and by
seed funding from the CyLab Security and Privacy Institute at Carnegie Mellon
University. The lead author was supported by a Sansom Endowed Presidential
Fellowship. The writing and insights in this paper were produced exclusively by
the minds of the authors, in the hope that others might benefit from them.

\vfill

\pagebreak

\bibliography{main}

@inproceedings{abadi1999,
author = {Abadi, Mart\'{\i}n and Banerjee, Anindya and Heintze, Nevin and Riecke, Jon G.},
title = {A core calculus of dependency},
year = {1999},
doi = {10.1145/292540.292555},
booktitle = {26th ACM SIGPLAN-SIGACT Symposium on Principles of Programming Languages},
series = {POPL},
}

@article{shikuma2008proving,
  author={Shikuma, Naokata and Igarashi, Atsushi},
  title={Proving Noninterference by a Fully Complete Translation to the Simply Typed lambda-calculus},
  year={2008},
  doi={10.2168/LMCS-4(3:10)2008},
  journal={Logical Methods in Computer Science},
  volume={4},
  number={3},
  pages={1--31}
}

@inproceedings{miyamoto2004modal,
  author={Miyamoto, Kenji and Igarashi, Atsushi},
  title={A Modal Foundation for Secure Information Flow},
  year={2004},
  url={https://www.fos.kuis.kyoto-u.ac.jp/~igarashi/papers/pdf/modal-FCS04.pdf},
  booktitle={Workshop on Foundations of Computer Security},
  series={FCS},
}

@article{liu2024internalizing,
  author={Liu, Yiyun and Chan, Jonathan and Shi, Jessica and Weirich, Stephanie},
  title={Internalizing Indistinguishability with Dependent Types},
  year={2024},
  doi={10.1145/3632886},
  journal={Proceedings of the ACM on Programming Languages},
  volume={8},
  number={POPL},
  pages={1298--1325},
}

@inproceedings{choudhury2022dependent,
  author={Choudhury, Pritam and Eades III, Harley and Weirich, Stephanie},
  title={A Dependent Dependency Calculus},
  year={2022},
  doi={10.1007/978-3-030-99336-8_15},
  booktitle={European Symposium on Programming},
  series={ESOP},
}

@inproceedings{reynolds_1984,
  author    = {Reynolds, John C.},
  title     = {{Types, Abstraction and Parametric Polymorphism}},
  year      = 1983,
  url = {https://people.mpi-sws.org/~dreyer/tor/papers/reynolds.pdf},
  booktitle = {Information Processing 83},
  series    = {IFIP Congress Series},
}

@article{torczon2024cbpv,
author = {Torczon, Cassia and Su\'{a}rez Acevedo, Emmanuel and Agrawal, Shubh and Velez-Ginorio, Joey and Weirich, Stephanie},
title = {Effects and Coeffects in Call-by-Push-Value},
year = {2024},
doi = {10.1145/3689750},
journal = {Proceedings of the ACM on Programming Languages},
volume = {8},
number = {OOPSLA2},
pages = {1108--1134},
}

@inproceedings{bowman2015noninterference,
author = {Bowman, William J. and Ahmed, Amal},
title = {Noninterference for Free},
year = {2015},
doi = {10.1145/2784731.2784733},
booktitle = {20th ACM SIGPLAN International Conference on Functional Programming},
series = {ICFP}
}

@inproceedings{myers1999mostly,
  author={Myers, Andrew C},
  title={JFlow: practical mostly-static information flow control},
  year={1999},
  doi={10.1145/292540.292561},
  booktitle={26th ACM SIGPLAN-SIGACT symposium on Principles of programming language},
  series={POPL},
}

@article{gouni2025structural,
author = {Gouni, Hemant and Pfenning, Frank and Aldrich, Jonathan},
title = {Structural Information Flow: A Fresh Look at Types for Non-interference},
year = {2025},
doi = {10.1145/3764116},
journal = {Proceedings of the ACM on Programming Languages},
volume = {9},
number = {OOPSLA2},
pages={3954--3980},
}

@inproceedings{moggi1989computational,
author = {Moggi, Eugenio},
title = {Computational lambda-calculus and monads},
year = {1989},
doi = {10.1109/LICS.1989.39155},
booktitle = {4th Annual Symposium on Logic in Computer Science},
series = {LICS},
}

@article{sterling2021logical,
  author={Sterling, Jonathan and Harper, Robert},
  title={Logical Relations as Types: Proof-Relevant Parametricity for Program Modules},
  year={2021},
  doi={10.1145/3474834},
  journal={Journal of the ACM},
  volume={68},
  number={6},
  pages={1--47},
}

@article{rijke2020modalities,
  author={Rijke, Egbert and Shulman, Michael and Spitters, Bas},
  title={Modalities in homotopy type theory},
  year={2020},
  doi={10.23638/LMCS-16(1:2)2020},
  journal={Logical Methods in Computer Science},
  volume={16},
  number={1},
  pages={1--79},
}

@inproceedings{sterling2022sheaf,
  author={Sterling, Jonathan and Harper, Robert},
  title={Sheaf Semantics of Termination-Insensitive Noninterference},
  year={2022},
  doi={10.4230/LIPIcs.FSCD.2022.5},
  booktitle={7th International Conference on Formal Structures for Computation and Deduction},
  series={FSCD},
}

@inproceedings{cecchetti2017nonmalleable,
  author={Cecchetti, Ethan and Myers, Andrew C and Arden, Owen},
  title={Nonmalleable Information Flow Control},
  year={2017},
  booktitle={ACM SIGSAC Conference on Computer and Communications Security},
  series={CCS},
  doi={10.1145/3133956.3134054}
}

@book{girard1989proofs,
    author = {Girard, Jean-Yves and Lafont, Yves and Taylor, Paul},
    title = {Proofs and {Types}},
    year = {1989},
    isbn = {9780521371810},
    series = {Cambridge {Tracts} in {Theoretical} {Computer} {Science}},
    publisher = {Cambridge University Press},
}

@article{choudhury2022monadic,
author = {Choudhury, Pritam},
title = {Monadic and Comonadic Aspects of Dependency Analysis},
year = {2022},
doi = {10.1145/3563335},
journal = {Proceedings of the ACM on Programming Languages},
volume = {6},
number = {OOPSLA2},
pages = {1320--1348},
}

@inproceedings{rajani2025graded,
  title={A Graded Modal Approach to Relaxed Semantic Declassification},
  author={Rajani, Vineet and Coleman, Alex and Kanabar, Hrutvik},
  year={2025},
  doi={10.1109/CSF64896.2025.00032},
  booktitle={38th Computer Security Foundations Symposium},
  series={CSF},
}

@inproceedings{sterling2021metalanguage,
  title={A metalanguage for multi-phase modularity},
  author={Sterling, Jonathan and Harper, Robert},
  year={2021},
  url={https://www.cs.cmu.edu/~rwh/papers/multiphase/mlw.pdf},
  booktitle={ML Workshop},
  series={MLW},
}

@article{hirsch2021giving,
  author={Hirsch, Andrew K and Cecchetti, Ethan},
  title={Giving semantics to program-counter labels via secure effects},
  year={2021},
  doi={10.1145/3434316},
  journal={Proceedings of the ACM on Programming Languages},
  volume={5},
  number={POPL},
  pages={1--29},
}

@book{levy2012call,
  title={Call-By-Push-Value: A Functional/Imperative Synthesis},
  author={Levy, Paul Blain},
  year={2003},
  doi={10.1007/978-94-007-0954-6},
  publisher={Springer Dordrecht}
}

@article{grodin2026abstraction,
  author={Grodin, Harrison and Li, Runming and Harper, Robert},
  title={Abstraction Functions as Types: Modular Verification of Cost and Behavior in Dependent Type Theory},
  year={2026},
  doi={10.1145/3776673},
  journal={Proceedings of the ACM on Programming Languages},
  volume={10},
  number={POPL},
  pages={895--922},
}

@Book{hottbook,
  author =    {The {Univalent Foundations Program}},
  title =     {Homotopy Type Theory: Univalent Foundations of Mathematics},
  publisher = {\url{https://homotopytypetheory.org/book}},
  address =   {Institute for Advanced Study},
  year =      2013
}

@INPROCEEDINGS{bugliesi2011resource,
  author={Bugliesi, Michele and Calzavara, Stefano and Eigner, Fabienne and Maffei, Matteo},
  title={Resource-Aware Authorization Policies for Statically Typed Cryptographic Protocols}, 
  year={2011},
  booktitle={24th Computer Security Foundations Symposium}, 
  series={CSF},
  doi={10.1109/CSF.2011.13}
}

@INPROCEEDINGS{gordon2001authenticity,
  author={Gordon, A.D. and Jeffrey, A.},
  title={Authenticity by typing for security protocols}, 
  year={2001},
  doi={10.1109/CSFW.2001.930143},
  booktitle={14th IEEE Computer Security Foundations Workshop}, 
  series={CSFW},
}

@article{swasey2017robust,
author = {Swasey, David and Garg, Deepak and Dreyer, Derek},
title = {Robust and compositional verification of object capability patterns},
year = {2017},
doi = {10.1145/3133913},
journal = {Proceedings of the ACM on Programming Languages},
volume = {1},
number = {OOPSLA},
pages = {1--26},
}

@ARTICLE{ieee1990standard,
  author={{IEEE Standards Board}},
  journal={IEEE Std 610.12-1990}, 
  title={IEEE Standard Glossary of Software Engineering Terminology}, 
  year={1990},
  volume={},
  number={},
  pages={1-84},
  doi={10.1109/IEEESTD.1990.101064}}

@article{mackay2022necessity,
author = {Mackay, Julian and Eisenbach, Susan and Noble, James and Drossopoulou, Sophia},
title = {Necessity Specifications for Robustness},
year = {2022},
doi = {10.1145/3563317},
journal = {Proceedings of the ACM on Programming Languages},
volume = {6},
number = {OOPSLA2},
pages = {811--840},
}

@article{mcbride2008applicative,
author = {McBride, Conor and Paterson, Ross},
title = {Applicative programming with effects},
year = {2008},
doi = {10.1017/S0956796807006326},
journal = {Journal of Functional Programming},
volume = {18},
number = {1},
pages = {1–13},
}

@inproceedings{zdancewic2001robust,
  title={Robust Declassification},
  author={Zdancewic, Steve and Myers, Andrew C},
  year={2001},
  doi={10.5555/872752.873524},
  booktitle={14th IEEE Computer Security Foundations Workshop},
  series={CSFW},
}

@article{orbaek1997trust,
  author={{\O}rb{\ae}k, Peter and Palsberg, Jens},
  title={Trust in the $\lambda$-calculus},
  year={1997},
  doi={10.1017/S0956796897002906},
  journal={Journal of Functional Programming},
  volume={7},
  number={6},
  pages={557--591},
}

@article{crary2020focused,
  author={Crary, Karl},
  title={A focused solution to the avoidance problem},
  year={2020},
  doi={10.1017/S0956796820000222},
  journal={Journal of Functional Programming},
  volume={30},
  number={Robert Harper Festschrift Collection},
}

@inproceedings{arden2016calculus,
  author={Arden, Owen and Myers, Andrew C},
  title={A Calculus for Flow-Limited Authorization},
  year={2016},
  doi={10.1109/CSF.2016.17},
  booktitle={29th Computer Security Foundations Symposium},
  series={CSF},
}

@article{egger2014enriched,
  author={Egger, Jeff and Ejlers, Rasmus and Simpson, Alex and others},
  title={The enriched effect calculus: syntax and semantics},
  year={2014},
  doi={10.1093/logcom/exs025},
  journal={Journal of Logic and Computation},
  volume={24},
  number={3},
  pages={615--654},
}

@techreport{biba1977integrity,
  title={Integrity Considerations for Secure Computer Systems},
  author={Biba, Kenneth J},
  year={1977},
  url={https://apps.dtic.mil/sti/citations/ADA039324}
}

@article{gouni2026security,
  author={Gouni, Hemant and Pfenning, Frank and Aldrich, Jonathan},
  title={Security Reasoning via Substructural Dependency Tracking},
  year={2026},
  doi={10.1145/3776669},
  journal={Proceedings of the ACM on Programming Languages},
  volume={10},
  number={POPL},
  pages={777--805},
}

@inproceedings{goguen1982security,
  author={Goguen, Joseph A and Meseguer, Jos{\'e}},
  title={Security Policies and Security Models},
  year={1982},
  doi={10.1109/SP.1982.10014},
  booktitle={IEEE Symposium on Security and Privacy},
  series={S\&P},
}

@book{schultz2017temporal,
  title={Temporal Type Theory: A Topos-Theoretic Approach to Systems and Behavior},
  author={Schultz, Patrick and Spivak, David I},
  doi={10.1007/978-3-030-00704-1},
  year={2019},
  publisher={Birkhäuser Cham},
}

\pagebreak

\appendix

\section{The Hybrid Modality}
\label{sec:hybrid}

\begin{figure}
\begin{mathpar}
    \inferrule[Mark]
    {\Gamma; \phi' \cup \phi \vdash e : A}
    {\Gamma; \phi' \vdash \texttt{box}\; e : T^{\phi} A}

    \inferrule[LetMark]
    {\Gamma; \phi' \vdash e_1 : T^{\phi}(A_1) \\ \Gamma, x :^{\phi} A_1; \phi' \vdash e_2 : A_2}
    {\Gamma; \phi' \vdash \texttt{unbox}\; e_1\; \texttt{as}\; x\; \texttt{in}\; e_2 : A_2}

    \inferrule[Var]
    {x :^{\phi} A \in \Gamma \\ \phi \subseteq \phi'}
    {\Gamma; \phi' \vdash x : A}
\end{mathpar}
\caption{Hybrid Modality}
\label{fig:hybrid}
\end{figure}

\autoref{fig:hybrid} shows a third kind of information flow modality appearing
in the literature \cite{miyamoto2004modal, choudhury2022dependent}. The
introduction rule \textsc{Mark} is structurally identical to that of the open
modality, reifying dependencies in the typing judgement within the modality
$T^\phi A$. However, \textsc{LetMark} is structurally closer to the elimination
rule for the closed modality; it inspects its first operand and passes it to
the second. Unlike \textsc{Unseal}, it does not operate by way of a sealing
judgement $A \geq \phi$, but by marking the bound variable $x$ with the
dependencies $\phi$ from the eliminated modality. Using $x$ will require the
judgemental dependencies $\phi'$ to be greater than the $\phi$ marking the
variable, in similar fashion to the eliminator \textsc{Produce} for the open
modality.

We are not aware of a particular justification for this formulation, and indeed
its metatheoretic standing is unclear. One might wonder, for instance, what
judgemental structure it internalizes. If the ambient dependencies $\phi$, one
would expect it to be formulated structurally as a function, as in the case of
the open modality, but it is not: \textsc{LetMark} is a positive, rather than
negative, elimination rule. Furthermore, there do not seem to be any
expressiveness benefits of this setup over the others. We suspect this `hybrid'
modality to simply be an accident of formalization, drafted when the precise
character of dependency tracking modalities was not as critical as it is here,
and set it aside.

\end{document}